\documentclass[11pt]{IEEEtran}
\onecolumn

\usepackage{amsmath,amsfonts,amssymb,amscd,amsthm,dsfont}
\usepackage{graphicx,epsfig}
\usepackage{color}
\usepackage{mathtools}
\usepackage{bbm}
\usepackage{url}
\usepackage{balance}
\usepackage{epstopdf}
\usepackage{subfigure,enumitem}
\usepackage{cite,url,hyperref}
\usepackage{verbatim}
\usepackage{bm}
\usepackage{multicol}
\usepackage{ulem}
\usepackage{booktabs}
\usepackage{multirow}

\def\msk{{\mathcal S_{n,k}}}
\def\msn{{\mathcal S_{n,0}}}
\def\mss{{\mathcal S_{n,{S[t]}}}}

\newcommand{\mP}{\mathbb{P}}
\newcommand{\mE}{\mathbb{E}}

\newcommand{\nn}{\nonumber}

\newtheorem{theorem}{\textbf{Theorem}}

\newtheorem{lemma}{\textbf{Lemma}}

\usepackage{color}

\allowdisplaybreaks

\begin{document}
%
\title{Quickest Anomaly Detection in Sensor Networks With Unlabeled Samples}
\author{Zhongchang Sun \quad Shaofeng Zou
	\thanks{ This paper was presented in part at the 2021 IEEE International Symposium on Information Theory \cite{sun2021dynamic}.}
	\thanks{Zhongchang Sun and Shaofeng Zou are with the Department of Electrical Engineering, University at Buffalo, Buffalo, NY 14228 USA (e-mail: \href{mailto:zhongcha@buffalo.edu}{zhongcha@buffalo.edu}, \href{mailto:szou3@buffalo.edu}{szou3@buffalo.edu}).}
}
\maketitle
\begin{abstract}
The problem of quickest anomaly detection in networks with unlabeled samples is studied.  At some unknown time, an anomaly emerges in the network and changes the data-generating distribution of some unknown sensor. The data vector received by the fusion center at each time step undergoes some unknown and arbitrary permutation of its entries (unlabeled samples). The goal of the fusion center is to detect the anomaly with minimal detection delay subject to false alarm constraints.
With unlabeled samples, existing approaches that combines local cumulative sum (CuSum) statistics cannot be used anymore. Several major questions include whether detection is still possible without the label information, if so, what is the fundamental limit and how to achieve that.  
Two cases with static and dynamic anomaly are investigated, where the sensor affected by the anomaly may or may not change with time. For the two cases, practical algorithms based on the ideas of mixture likelihood ratio and/or maximum likelihood estimate are constructed. Their average detection delays and false alarm rates are theoretically characterized. Universal lower bounds on the average detection delay for a given false alarm rate are also derived, which further demonstrate the asymptotic optimality of the two algorithms. 
\end{abstract}

\section{Introduction}

In large-scale sensor networks, samples may lack of label information such as identity due to, e.g.,  malicious attacks and limited communication resources. For example, 
wireless ad-hoc sensor networks are usually vulnerable to spoofing attacks \cite{humphreys2008assessing}, and samples received by the fusion center may then lose their label information. Furthermore, in large-scale Internet-of-things (IoT) networks, where devices are commonly small and low-cost sensing devices powered by battery with limited communication bandwidth, and are usually deployed in a massive scale, the communication overhead of identifying individual sensors increases drastically as the number of sensors grows \cite{keller2009identity}. However, these battery-powered IoT devices are usually expected to survive for years without battery change. In this case, message delivered to the fusion center may be constrained not to contain the identity information. For the same reason, the data transmitted to the fusion center is usually quantized to be in a finite alphabet. Furthermore, in social sensing applications, participants may choose to be anonymous in order to protect privacy, i.e., sharing the data without including identity information. Motivated by these applications, there is a recent surge of interest in the problem of signal processing with unlabeled data (see e.g., \cite{anonymous,li2022bandwidth,stefano2019unlabeled,stefano2020bits,unnikrishnan2018unlabeled,haghighatshoar2017signal,abid2017linear,emiya2014compressed,liu2018signal,pananjady2017linear,elhami2017unlabeled,lu2008theory,wang2018signal,sun2020tspanonymous}), which refers to various signal processing problems where the data vector undergoes an unknown permutation of its entries, and the original position of each datum in the vector is unknown.

In this paper, we investigate the problem of quickest anomaly detection in sensor networks with unlabeled samples. Specifically, at some unknown time, an anomaly emerges in the network and leads to a change in the data-generating distribution of some unknown sensor. The fusion center sequentially receives unlabeled (arbitrarily permuted) samples from all the sensors at each time step. The goal of the fusion center is to detect the anomaly as quickly as possible, subject to false alarm constraints. This problem is of particular relevance to applications where an anomaly affects some sensor in the network, and the affected sensor may change over time \cite{rovatsos2021quickest}, e.g., surveillance system, intrusion detection, environmental change (air/water quality) detection, rumor detection, and seismic wave detection.

\subsection{Contributions and Major Challenges}
The first part of this paper focuses on the static anomaly, where the sensor affected by the anomaly does not change with time, but which sensor is affected is still unknown. We consider the detection delay under the worst-case affected sensor. The goal here is to minimize the detection delay subject to false alarm constraints. 
The major challenges here are two-fold. First of all, the labels of the samples are unknown, and is time-varying. Second, even if the labels are known, i.e., each sample is associated with its sensor,  the sensor the anomaly affects is still unknown.
For this problem, we construct a generalized mixture CuSum (GM-CuSum) algorithm. The basic idea is to estimate the unknown identity of the affected sensor using the maximum likelihood estimate (MLE), and further employ a mixture likelihood w.r.t. all possible labels. We prove that the GM-CuSum is second-order asymptotically optimal. 

The second part of this paper focuses on a general and more challenging setting with dynamic anomaly, where the sensor affected by the anomaly changes with time. Here, we refer to the sequence of sensors affected by the anomaly over time as the trajectory of the anomaly. We consider the detection delay under the worst-case trajectory. 
Compared to the static setting, the additional challenge is that the affected sensor changes with time, and thus the change is not persistent at any particular sensor. Therefore, estimating the identity of the affected sensor over time is not applicable. We then propose a Bayesian approach to address the challenge raised by the unknown trajectory of the anomaly, and find the optimal weight to construct a weighted mixture CuSum algorithm. We prove that the weighted mixture CuSum algorithm is first-order asymptotically optimal. 

We also conduct extensive numerical results to demonstrate the performance of our proposed algorithms. The numerical results show that for the static setting, our GM-CuSum algorithm outperforms a heuristic Bayesian mixture CuSum algorithm; the optimal weighted mixture CuSum algorithm also performs well for the static setting; and for the dynamic setting, our optimal weighted mixture CuSum algorithm outperforms an arbitrarily weighted one and the GM-CuSum algorithm. These numerical results validate our theoretical optimality results.
\subsection{Related Work}
The quickest change detection (QCD) problem in sensor networks with labeled samples was extensively  studied  in the literature, e.g., \cite{tartakovsky2004change,tartakovsky2006novel,mei2010efficient,xie2013sequential,fellouris2016second,raghavan2010quickest, hadjiliadis2009one,ludkovski2012bayesian,zou2020dynamic, veeravalli2001decentralized,tartakovsky2008distributed,zou2019distributed,xie2021sequential} where the fusion center knows the identity of each sample, i.e., knows which sensor that each sample is from. Therefore, one CuSum algorithm can be implemented at each sensor and then be combined to make the decision. This type of algorithms were shown to be asymptotically optimal for various settings. In this paper, we investigate the setting with unlabeled samples, where at each time step samples are arbitrarily permuted, and the permutation is time-varying. The fusion center does not know which sensor each sample comes from, and then cannot implement a CuSum algorithm for each sensor. 

Various learning and inference problems with unlabeled data has been studied in the literature \cite{anonymous,li2022bandwidth,stefano2019unlabeled,stefano2020bits,unnikrishnan2018unlabeled,haghighatshoar2017signal,abid2017linear,emiya2014compressed,liu2018signal,pananjady2017linear,elhami2017unlabeled,lu2008theory,wang2018signal,sun2020tspanonymous}, which mainly focus on the offline setting with non-sequential data. Here we only review several closely related ones on detection problems. 
In \cite{stefano2019unlabeled}, hypothesis testing with unlabeled samples are studied, where two practical algorithms, the unlabeled log-likelihood ratio test and the generalized likelihood ratio test are proposed. A more specific problem is studied in \cite{stefano2020bits} where samples follow Bernoulli distribution and an approximated log-likelihood test based on the central limit theorem was proposed. In \cite{anonymous}, the binary hypothesis testing problem with unlabeled samples was studied, and an optimal mixture likelihood ratio test (MLRT) was developed. 
In \cite{li2022bandwidth}, the bandwidth-constrained QCD problem with unlabeled samples was investigated, where each sensor sends 1-bit quantized feedback to the fusion center. 
In \cite{sun2020tspanonymous}, the QCD problem with unlabeled samples was studied where the change affects all the sensors simultaneously. 
In this paper, we investigate a more practical scenario where an anomaly may not affect all the sensors, which is of particular interest in the distributed setting, and the anomaly may also be dynamic, and affect different sensors at different times, e.g., a moving target in surveillance system. 

Existing studies of quickly detecting a dynamic change mostly focus on the labeled setting, e.g., \cite{zou2020dynamic,rovatsos2019dynamic,georgios2020movinganomaly}. 
Our problem is similar to the one in \cite{georgios2020movinganomaly} but we focus on unlabeled samples. Our major technical challenge is due to the additional ambiguity of unknown labels.
The QCD problem with a slowly changing post-change distribution was studied in \cite{zou2018quickest,zhang2019quickest}, whereas in this paper, the anomaly can move arbitrarily fast.

With unlabeled samples, our problem is also related to the composite QCD problem with unknown pre- and post-change parameters e.g., \cite{siegmund1995using, lai1998information, xie2013sequential, banerjee2015composite}. Our work is different from the existing literature. Due to unlabeled samples and the dynamic nature of the anomaly,  the unknown parameter, i.e., the identify and the label of the affected sensor, is time-varying.  Therefore, the generalized likelihood approach which estimates the unknown parameters using their MLEs may not perform well. Moreover, unlike studies in \cite{siegmund1995using, lai1998information, banerjee2015composite} where the distributions are assumed to belong to the exponential family, we do not have any assumptions on the distributions.

%


\section{Problem formulation}\label{sec:problemmodel}
Consider a network monitored in real time by a set of $n$ heterogeneous sensors. These sensors can be clustered into $K$ types, and each type $k$ has $n_k$ sensors, $1\leq k\leq K$. The data generating distributions of samples from type $k$ sensors are denoted by $p_{\theta,k}$,  $\theta\in \{0,1\}$, which are known to the fusion center. At some unknown time $\nu$, an anomaly emerges in the network, and changes the data-generating distributions of the sensors. If a sensor of type $k$ is affected by the anomaly, then its samples are generated by $p_{1,k}$, otherwise, by $p_{0,k}$. The goal is to detect the anomaly as quickly as possible subject to false alarm constraints. We focus on the case with unlabeled samples, where 
the data vector  at each time step undergoes an unknown permutation of its entries, and the original position of each datum in the vector is unknown to 
the fusion center. In other words, the fusion center does not know which type of sensors that each sample comes from, and therefore, does not know the sample's exact data-generating distribution. 

%
%
%

After the anomaly emerges, one sensor of an unknown type is affected by the anomaly. 
Based on whether the sensor is affected by the anomaly and the type of the sensor, we rearrange the sensors into $2K$ groups.  The first $K$ groups consists of sensors that are not affected by the anomaly; and the second $K$ groups consists of sensors that are affected by the anomaly. Specifically, for sensors in group $1\leq k\leq K$, their samples are generated by $p_{0,k}$, and for sensors in group  $K<k\leq 2K$, their samples samples are generated by $p_{1,k-K}$.

Denote by $X^n[t]=\{X_1[t],\ldots,X_n[t]\}$ the $n$ arbitrarily permuted samples at time $t$ received by the fusion center. We assume that $X_1[t],\ldots,X_n[t]$ are independent, and $X^n[{t_1}]$ is independent from $X^n[{t_2}]$ for any $t_1\neq t_2$.  Note that $X_i[t]$ is not necessarily the sample from sensor $i$ since samples are permuted/unlabeled. 

Let $\mathcal{K} = \{1,2,\cdots,K\}$. Denote by $S[t]\in \mathcal{K}\cup\{0\}$ the type of the affected sensor at $t$. 
For notational convenience, we use $S[t] = 0$ to denote the case when there is no anomaly, i.e., $t<\nu$. Let $\bm{S} \triangleq \{S[t]\}^\infty_{t=1}$ denote the trajectory of the anomaly. Here $\bm{S}$ is \textit{unknown} to the decision maker.
Even if the trajectory of the anomaly $\bm{S}$ is given, the distribution of $X^n[t]$ still cannot be fully specified due to lack of label information. To characterize the distribution of $X^n[t]$, we define a label function $\sigma_t^{S[t]}:\{1,\ldots,n\}\rightarrow \{1,\ldots,K,S[t]+K\}$. This function associates sample $X_i[t]$, $1\leq i\leq n$, to group $j$ for some $j\in\{1,2,\ldots,K,K+S[t] \}$, i.e., specifies the probability distribution of $X_i[t]$. Specifically, if $\sigma_t^{S[t]} (i)=j$, then
\begin{flalign}
	X_i[t]\sim \left\{
	\begin{array}{ll}p_{0,j}, &\text{ if } 1\leq j\leq K,\\
		p_{1,j-K}, &\text{ if } K<j\leq 2K.\end{array}
	\right.
\end{flalign}
Here $\sigma_t^{S[t]}$ can be interpreted as the inverse of the permutation applied to the data vector.
We further note that $\sigma_t^{S[t]}$ is \textit{unknown} to the decision maker, and changes with time.
%

Let $\Omega_{\bm{S}} = \{\sigma^{S[1]}_1,...,\sigma^{S[\infty]}_\infty\}$ be the labels when the trajectory of the anomaly is ${\bm{S}}$, which is unknown. Let $\mP^{{\bm{S}},\nu}_{\Omega_{\bm{S}}}$ and $\mathbb{E}^{{\bm{S}},\nu}_{\Omega_{\bm{S}}}$ denote the probability measure and the corresponding expectation when the change point is at $\nu$ and the samples received by the fusion center is permuted according to the  label $\Omega_{\bm{S}}$ (see Appendix \ref{app:a} for more details). We further let $\mP^{\infty}_{\Omega}$ and $\mathbb{E}^{\infty}_{\Omega}$ denote  the probability measure and the corresponding expectation when there is no change, i.e., $\nu=\infty$, where $\Omega=\Omega_{\bm{S}}$ with $S[t]=0, \forall t\geq1$.

We extend Lorden's criterion \cite{lorden1971procedures}, and define the worst-case average detection delay (WADD) and the worst-case average running length (WARL) for any stopping time $\tau$:
\begin{flalign}
	\text{WADD}(\tau) 
	&= \sup_{\nu \geq 1}\sup_{\bm{S}}\sup_{\Omega_{\bm{S}}}\text{esssup}\mE^{{\bm{S}},\nu}_{\Omega_{\bm{S}}}[(\tau-\nu)^+|\mathbf X^n[1,\nu-1]],\nn\\
	 \text{WARL}(\tau) &= \inf_\Omega\mE^\infty_\Omega[\tau],\label{eq:warl}
\end{flalign}
where $\mathbf X^n[t_1,t_2] = \{X^n[t_1],\cdots, X^n[t_2]\}$, for any $t_1\leq t_2$.
The goal is to design a stopping rule that minimizes the WADD subject to a constraint on the WARL:
\begin{flalign}\label{eq:sgoal}
	\inf_{\tau:\text{WARL}(\tau)\geq \gamma} \text{WADD}(\tau),
\end{flalign}
where $\gamma>0$ is a pre-specified threshold. Here the false alarm constraint is to guarantee that under all possible sample permutations, the average running length to a false alarm is always lower bounded by $\gamma$, and $1/\gamma$ can be interpreted as the false alarm rate.

\section{Static Anomaly}\label{sec:static}
We first investigate the case with static anomaly, i.e., the sensor affected by the anomaly does not change with time. In this case, for any $t\geq \nu$, $S[t]=k$ for some unknown type $k$.
Then, for all $j\in\{1,2,\cdots,K,k+K\}$, there are $\left(\substack{n\\n_1,\ldots, n_{k}-1, \ldots, n_K, 1}\right)$ possible $\sigma_t^{k}$ to associate each sample with a data-generating distribution,
and we denote the collection of all possible labels by $\mathcal S_{n,{k}}$ (see Appendix \ref{app:a} for more details).
Before the anomaly emerges, i.e., $t<\nu$, the samples $X^n[t]$ follows the distribution
\begin{flalign}\label{eq:kp_0}
\mP_{0,\sigma_t^0}(X^n[t])\overset{\Delta}{=}\prod_{i=1}^np_{0,\sigma_t^0(i)}(X_i[t]),
\end{flalign}  for some unknown $\sigma_t^0\in\mathcal S_{n,0}$. At time $t\geq \nu$, $S[t] = k$, $X^n[t]$ follows the distribution 
\begin{flalign}\label{eq:kp_1}
\mP^{k}_{\sigma^{k}_t}(X^n[t])\overset{\Delta}{=}&\prod\limits_{\substack{i: \sigma_t^k(i)\leq K}}p_{0,\sigma^{k}_t(i)}(X_i[t])\times\prod\limits_{\substack{i:\sigma_t^k(i)> K}}p_{1,\sigma^{k}_t(i)-K}(X_i[t]) ,
\end{flalign}
for some unknown $\sigma^{k}_t\in\mathcal S_{n,{k}}$.
Let $\Omega_k = \{\sigma_1^0,\ldots,\sigma_{\nu-1}^0,\sigma_\nu^k,\ldots,\sigma_\infty^k\}$ be the  labels over time, when the anomaly emerges at $\nu$ (similarly defined as $\Omega_{\bm{S}}$).  Let $\mP^{k,\nu}_{\Omega_k}$ denote the probability measure when the change point is at $\nu$ and the samples are generated according to \eqref{eq:kp_0}, \eqref{eq:kp_1} and $\Omega_k$. We further let $\mathbb{E}^{k,\nu}_{\Omega_k}$ denotes the corresponding expectation.

Then, the WADD for any stopping time $\tau$ can be written as
\begin{flalign}
	&\text{WADD}(\tau)=\sup_{\nu \geq 1}\sup_k\sup_{\Omega_k}\text{esssup}\mE^{k,\nu}_{\Omega_k}[(\tau-\nu)^+|\mathbf X^n[1,\nu-1]].\nn
\end{flalign}
The WARL is defined in the same way as in \eqref{eq:warl}.

The goal is to design a stopping rule that minimizes the $\text{WADD}$ subject to a constraint on the WARL:
\begin{flalign}\label{eq:kgoal}
\inf_{\tau:\text{WARL}(\tau)\geq \gamma} \text{WADD}(\tau).
\end{flalign}

\subsection{Universal Lower Bound on WADD}
We first derive a universal lower bound on $\text{WADD}$ for any $\tau$ satisfying the false alarm constraint: $\inf_\Omega\mE^\infty_\Omega[\tau]\geq \gamma$. 

Let $I_k = D(\widetilde{\mP}_k||\widetilde{\mP}_0)$ denote the Kullback-Leibler (KL) divergence between two mixture distributions $\widetilde{\mP}^k=\frac{1}{\mid \msk\mid}\sum_{\sigma\in \msk}\mP^{k}_{\sigma}$ and $\widetilde{\mP}_0=\frac{1}{\mid \msn\mid}\sum_{\sigma\in \msn}\mP_{0,\sigma}$. 
Here, $\widetilde{\mP}^k$ is the uniform mixture of all possible distributions when the affected sensor is from group $k$.
Let $I^* = \min_{1\leq k \leq K}I_k$. We then have the following theorem.
\begin{theorem}\label{theorem:1}
	As $\gamma \rightarrow \infty$, 
	\begin{flalign}
		\inf_{\tau:\text{WARL}(\tau)\geq \gamma} \text{WADD}(\tau) \geq \frac{\log\gamma}{I^*}+O(1).
	\end{flalign}
\end{theorem}
The proof of Theorem \ref{theorem:1} can be found in Appendix \ref{sec:kaddlow}.
The main challenges in the proof of Theorem \ref{theorem:1} is due to the worst-case over all labels and affected sensors in WADD and WARL.  From Theorem \ref{theorem:1}, it can be seen that the WADD for problem \eqref{eq:kgoal} is lower bounded by $\frac{\log\gamma}{I^*}+O(1)$ for any stopping rule that satisfy the constraint on WARL. 
Theorem \ref{theorem:1} motivates us to find the $k$ that minimizes $I_k$, i.e., achieves $I^*$, and design an algorithm to achieve this universal lower bound.


\subsection{Generalized Mixture CuSum Algorithm}
In this section, we construct an algorithm that achieves the universal lower bound asymptotically. 

A first idea is to use the MLE to estimate the unknown label $\sigma_t^k$ and the unknown affected sensor $k$. In the static setting, $k$ does not change with time, however, $\sigma_t^k$ changes with time, thus a direct MLE for $\sigma_t^k$ at each time $t$ may not work well. 
Therefore, we take a mixture approach w.r.t. the unknown label, and then take a MLE approach w.r.t. the unknown affected sensor. Our algorithm is constructed as follows. 

Let $W[t] = \max_{k\in\mathcal{K}}\max_{1\leq j \leq t}\sum^t_{i = j}\log\frac{\widetilde{\mP}^k(X^n[i])}{\widetilde{\mP}_0(X^n[i])}$. 
We then define the GM-CuSum stopping time as follows:
\begin{flalign}\label{eq:tg}
	T_G = \inf\{t:W[t]\geq b\},
\end{flalign} 
where $b>0$ is the threshold. Here $W[t]$ can be updated efficiently. We keep $K$ CuSums in parallel. Note that this can be done recursively. Let $W_k[t] = \max_{1\leq j \leq t}\sum\limits^t_{i = j}\log\frac{\widetilde{\mP}^k(X^n[i])}{\widetilde{\mP}_0(X^n[i])}$. The test statistic $W[t]$ has the following recursion:
\begin{flalign}
&W[t+1] = \max_{k\in \mathcal{K}}\bigg\{(W_k[t])^+ +  \log\frac{\widetilde{\mP}^k(X^n[t+1])}{\widetilde{\mP}_0(X^n[t+1])}\bigg\},
\end{flalign}
where $W_k[0] = 0$ for any $k\in \mathcal{K}$.  We then take their maximum over $k$ as $W[t]$. 

In the following, we show 1) the WARL lower bound of $T_G$ and 2) the WADD upper bound of $T_G$ in the following theorem. 
\begin{theorem}\label{THEOREM:2}
	1) Let $b = \log ({K\gamma})$ in \eqref{eq:tg}. Then $\text{WARL}(T_G)\geq \gamma$; and 
	2) As $\gamma \rightarrow \infty$, 
	$
		\text{WADD}(T_G) \leq \frac{\log \gamma}{I^*} + O(1).
	$
\end{theorem}
	
The proof of Theorem \ref{THEOREM:2} can be found in Appendix \ref{sec:karlproof}.

The proof of the lower bound on WARL is based on Doob's submartingale inequality \cite{williams_1991} and the optional sampling theorem \cite{williams_1991}. The major challenge lies in that we consider the worst-case label. A key property we develop and use in the proof of the WARL lower bound is that under the pre-change distribution $\mP_{0,\sigma_t^0}$, for any $k\in\mathcal{K}$, the expectation of the mixture likelihood ratio $\mE_{{0,\sigma^0}}\Big[\log\frac{\widetilde{\mP}^k(X^n)}{\widetilde{\mP}_0(X^n)}\Big]$ is invariant for different $\sigma^0$'s. 

Theorem $\ref{THEOREM:2}$ suggests that to meet the WARL constraint, $b$ should be chosen such that $b = \log {K\gamma}$.

Based on Theorem \ref{theorem:1} and Theorem \ref{THEOREM:2}, we then establish the second-order asymptotic optimality of $T_G$ in the following theorem.
\begin{theorem}\label{theorem:Gmopt}
$T_G$ is second-order asymptotically optimal for the problem in \eqref{eq:kgoal}.
\end{theorem}
\begin{proof}
By Theorem $\ref{theorem:1}$ and Theorem $\ref{THEOREM:2}$, we establish the second-order asymptotic optimality of $T_G$.
\end{proof}

\section{Quickest Dynamic Anomaly Detection}\label{sec:dynamic}
In this section, we consider the general problem with a dynamic anomaly, where the sensor affected by the anomaly changes with time. The GM-CuSum algorithm designed for static anomaly may not work well anymore since the sensor affected by the anomaly changes with time.
\subsection{Universal Lower Bound on WADD}

Define the following weighted mixture distribution:
$
	\widetilde{\mP}^{\bm{\beta}}(X^n) = \sum_{k=1}^{K}\beta_k\widetilde{\mP}^k(X^n),
$
where $\bm{\beta} = \{\beta_k\}_{k=1}^K$, $0\leq \beta_k \leq 1$ and $\sum_{k=1}^K\beta_k = 1$.
Denote by $I_{\bm{\beta}}$ the KL divergence between $\widetilde{\mP}^{\bm{\beta}}$ and $\widetilde{\mP}_0$. 
Let $\bm{\beta^*} = \mathop{\arg\min}_{\bm{\beta}}I_{\bm\beta}$.

For the universal lower bound on WADD, we have the following theorem.
\begin{theorem}\label{theorem:3}
	As $\gamma \rightarrow \infty$, we have that 
	\begin{flalign}\label{eq:lowadds}
		\inf_{\tau:\text{WARL}(\tau)\geq \gamma}\text{WADD}(\tau)\geq \frac{\log\gamma}{I_{\bm{\beta^*}}}(1+o(1)).
	\end{flalign}
\end{theorem}
The proof of Theorem \ref{theorem:3} can be found in Appendix \ref{sec:lowwadds}.

It can be seen from Theorem \ref{theorem:3} that the WADD for the problem in \eqref{eq:sgoal} is lower bounded by $\frac{\log\gamma}{I_{\bm{\beta^*}}}(1+o(1))$ for large $\gamma$. This motivates us to apply the optimal weight $\bm{\beta^*}$ to design an algorithm that can achieves the WADD lower bound asymptotically.
Moreover, we have that $I^*\geq I_{\bm{\beta^*}}$ which implies that a dynamic anomaly is more difficult to detect than a static anomaly.

\subsection{Weighted Mixture CuSum}
In the static setting, the unknown affected sensor can be estimated by its MLE. However, in the dynamic setting, the affected sensor changes with time, and the MLE approach may not work well. Theorem \eqref{theorem:3} motivates us to tackle the unknown anomaly trajectory using a Bayesian approach where the probability that the $k$-th group is affected by the anomaly is $\beta_k^*$. We then construct our weighted mixture CuSum algorithm as follows. 
Define the log of weighted mixture likelihood ratio using $\bm{\beta^*}$:
\begin{flalign}
	\ell_{\bm{\beta^*}}(X^n)=
	\log \frac{\widetilde{\mP}^{\bm{\beta^*}}(X^n)}{\widetilde{\mP}_0(X^n)}.
\end{flalign}
It can be easily shown that $\ell_{\bm{\beta^*}}(X^n)$ is invariant to any permutations on $X^n$, i.e.,  for any permutation $\pi(X^n)=(X_{\pi(1)},X_{\pi(2)},\ldots,X_{\pi(n)})$, $\ell_{\bm{\beta^*}}(X^n)=\ell_{\bm{\beta^*}}(\pi(X^n))$. This is due to the fact that  $\ell_{\bm{\beta^*}}(X^n)$ takes the sum over all possible group assignments thus is invariant to the actual permutation of samples.

We then construct the following weighted mixture CuSum algorithm:
\begin{flalign}\label{eq:wmixture}
	T_{\bm{\beta^*}}(b) = \inf\Big\{t: \max\limits_{1\leq j\leq t+1}\sum_{i=j}^t \ell_{\bm\beta^*}(X^n[i])\geq b\Big\}.
\end{flalign}
Let $\widehat{W}[t] = \max\limits_{1\leq j\leq t+1}\sum_{i=j}^t \ell_{\bm\beta^*}(X^n[i])$. The test statistic $\widehat{W}[t]$ has the following recursion: 
$
\widehat{W}[t+1] = (\widehat W[t])^+ + \ell_{\bm\beta^*}(X^n[t+1]), \widehat{W}[0] = 0.
$

Note that different from the way that we handle the unknown and time-varying label $\sigma$, here, for the unknown type of the affected sensor, we take the mixture according to $\beta^*$ instead of  a uniform distribution over $\mathcal K$. As will be shown later both theoretically in Theorem  \ref{theorem:WMopt} and numerically in Section \ref{sec:numerical}, taking a uniform mixture over $\mathcal K$ may not lead to the optimal performance. 

Let $\widetilde{\mE}^{k}$ and $\widetilde{\mE}_{0}$ denote the expectation under the probability $\widetilde{\mP}^{k}$ and $\widetilde{\mP}_{0}$ respectively.
The following property of $\bm{\beta^*}$ plays an important role in developing the asymptotic optimality of the weighted mixture CuSum algorithm.  
\begin{lemma}\label{lemma:1}
	For any $k\in\mathcal{K}$,
	\begin{flalign}
		\widetilde{\mE}^{k}\Big[\log\frac{\widetilde{\mP}^{\bm{\beta^*}}(X^n)}{\widetilde{\mP}_0(X^n)}\Big] \geq I_{\bm{\beta^*}}.
	\end{flalign}
\end{lemma}
The proof of Lemma \ref{lemma:1} can be found in Appendix \ref{sec:lemma}.

In the following, we provide a heuristic explanation of how $\widehat{W}[t]$ evolves in the pre- and post-change regimes. 
We first argue that  $\mE^{k}_{\sigma^k}[\ell_{\bm{\beta^*}}(X^n)]$ is invariant for different $\sigma^k$'s. Specifically, let $\mE^{k}_{\sigma^k}$ denote the expectation under $\mP^{k}_{\sigma^k}$, where a sensor of type $k$ is affected, and the data received is labeled according to $\sigma^k$. For any $\pi$, let $\hat{\sigma}^k=\sigma^k\circ\pi$. Then $\mE^{k}_{\sigma^k}[\ell_{\bm{\beta^*}}(\pi(X^n))]=\mE^{k}_{\sigma^k\circ\pi}[\ell_{\bm{\beta^*}}(X^n)]=\mE^{k}_{\hat{\sigma}^k}[\ell_{\bm{\beta^*}}(X^n)]$. For any $\hat{\sigma}^k\in\msk$, a $\pi$ can always be found so that $\sigma^k\circ\pi=\hat\sigma^k$. Thus, for any $\sigma^k,\hat\sigma^k\in\msk$,
$\mE^{k}_{\hat\sigma^k}[\ell_{\bm{\beta^*}}(X^n)]=\mE^{k}_{\sigma^k}[\ell_{\bm{\beta^*}}(X^n)]$.
Therefore, $\mE^{k}_{\sigma^k}[\ell_{\bm{\beta^*}}(X^n)]$ is invariant for different $\sigma^k$'s. 
Then, under the pre-change distribution $\mP_{0,\sigma_t^0}$, the expectation of the weighted mixture likelihood ratio $\mE_{{0,\sigma_t^0}}[\ell_{\bm{\beta^*}}(X^n)]$ is invariant for different $\sigma^0_t$'s, we have that 
\begin{flalign}
&\mE_{{0,\sigma_t^0}}\Big[\log\frac{\widetilde{\mP}^{\bm{\beta^*}}(X^n)}{\widetilde{\mP}_0(X^n)}\Big]\nn\\& = \frac{1}{\mid \msn\mid}\sum_{\sigma^0_t\in \msn}\mE_{ {0,\sigma_t^0}}\Big[\log\frac{\widetilde{\mP}^{\bm{\beta^*}}(X^n)}{\widetilde{\mP}_0(X^n)}\Big]\nn\\& = \widetilde{\mE}_{0}\Big[\log\frac{\widetilde{\mP}^{\bm{\beta^*}}(X^n)}{\widetilde{\mP}_0(X^n)}\Big]\nn\\& = -D(\widetilde \mP_0||\widetilde{\mP}^{\bm{\beta^*}}) \leq0. 
\end{flalign}
Therefore, before the change time $\nu$, $\widehat{W}[t]$ has a negative drift. Similarly, from Lemma \ref{lemma:1}, after the change time $\nu$, under any group assignment $\Omega_{\bm{S}}$ and trajectory $\bm{S}$, $\widehat{W}[t]$ has a positive drift whose expectation is no less than $I_{\bm{\beta^*}}$, and evolves towards $\infty$.


The following theorem establishes 1) the WARL lower bound of $T_{\bm{\beta^*}}$, and 2) the WADD upper bound of $T_{\bm{\beta^*}}$.
\begin{theorem}\label{theorem:4}
	1) For $T_{\bm{\beta^*}}$ defined in \eqref{eq:wmixture}, let $b = \log\gamma$, then $\text{WARL}(T_{\bm{\beta^*}})\geq \gamma$.
	
	2) As $\gamma\rightarrow \infty$, we have that 
	\begin{flalign}\label{eq:upwadds}
		\text{WADD}(T_{\bm{\beta^*}})  \leq \frac{\log\gamma}{I_{\bm{\beta^*}}}(1+o(1)).
	\end{flalign}
\end{theorem}
%
The proof of Theorem \ref{theorem:4} can be found in Appendix \ref{sec:upperwadds}.
The proof of Theorem \ref{theorem:4} is based on the Weak Law of Large Numbers for the weighted mixture likelihood ratio similarly to \cite{lai1998information}. The major challenge lies in that here we are interested in the worst-case label and the worst-case anomaly trajectory. Note that in our problem, the label and the affected sensor change with time. Therefore, it's challenging to explicitly characterize the worst-case label and anomaly trajectory for $T_{\bm{\beta^*}}$. To show the asymptotically optimal performance of $T_{\bm{\beta^*}}$, instead of finding the worst-case label and anomaly trajectory, we apply the symmetric property of $T_{\bm{\beta^*}}$ and Lemma \ref{lemma:1} to show that the WADD and WARL of $T_{\bm{\beta^*}}$ are bounded under all possible labels and trajectories.

We then establish the first-order asymptotic optimality of $T_{\bm\beta^*}$ in the following theorem.
\begin{theorem}\label{theorem:WMopt}
$T_{\bm\beta^*}$ is first-order asymptotically optimal for problem \eqref{eq:sgoal}.
\end{theorem}
\begin{proof}
Combining Theorem \ref{theorem:3} and Theorem \ref{theorem:4}, we establish the first-order asymptotic optimality of $T_{\bm\beta^*}$.
\end{proof}
If we apply $T_{\bm{\beta^*}}$ (designed for the dyanmic setting) to the static setting, the WADD of $T_{\bm{\beta^*}}$ can also be upper bounded by $\frac{\log\gamma}{I_{\bm{\beta^*}}}(1+o(1))$. However, $T_{\bm{\beta^*}}$ may not be asymptotically optimal. 
On the other hand, in the dynamic setting, the sensor affected by the anomaly changes with time, and thus the MLE may not work well. Therefore, the weighted mixture CuSum algorithm works better than the GM-CuSum.

\section{Simulation Results}\label{sec:numerical}

We first consider the static setting. We show an example evolution path of the GM-CuSum algorithm. We set $n=2$ and $K=2$. For type \uppercase\expandafter{\romannumeral1} sensors, the pre- and post-change distributions are $\mathcal B(10,0.3)$ and $\mathcal B(10,0.4)$, respectively, where $\mathcal B$ denotes binomial distribution. For type \uppercase\expandafter{\romannumeral2} sensors, the pre- and post-change distributions are $\mathcal B(10,0.8)$ and $\mathcal B(10,0.6)$, respectively. We set the change point to be 500 and $b=20$. We plot one sample evolution path of the GM-CuSum algorithm when one sensor of type one is affected. It can be seen from Fig.~\ref{fig:1} that before the change point, the test statistic fluctuates around zero, and after the change point, it starts to increase with a positive drift. 


\begin{figure*}[!t]
	\begin{multicols}{4}
		\includegraphics[width=\linewidth]{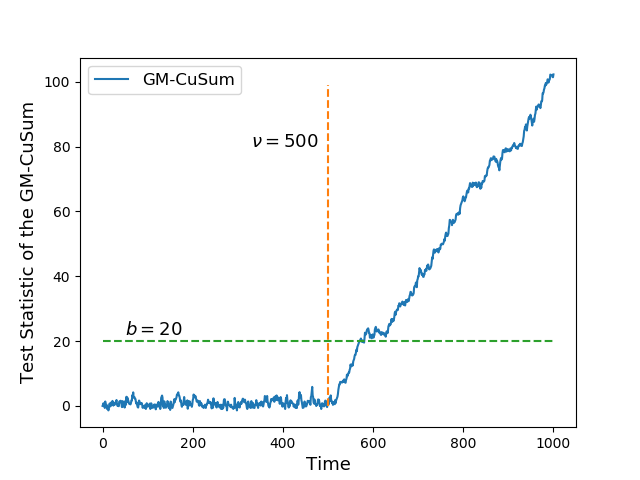}\par\caption{Evolution path of the GM-CuSum algorithm.}\label{fig:1}
		\includegraphics[width=\linewidth]{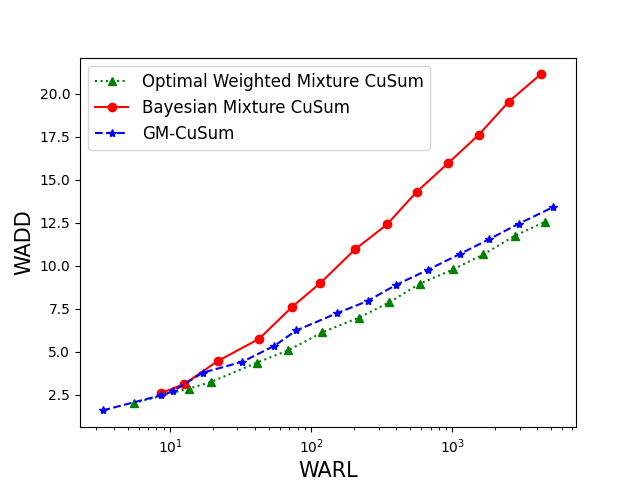}\par\caption{Comparison of the three algorithms in static setting: $n=4, K = 2$.}\label{fig:42}
		\includegraphics[width=\linewidth]{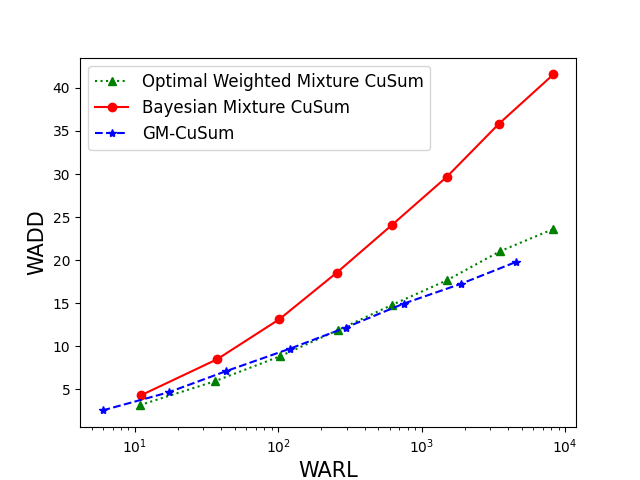}\par\caption{Comparison of the three algorithms in static setting: $n=8, K = 2$.}\label{fig:82}
		\includegraphics[width=\linewidth]{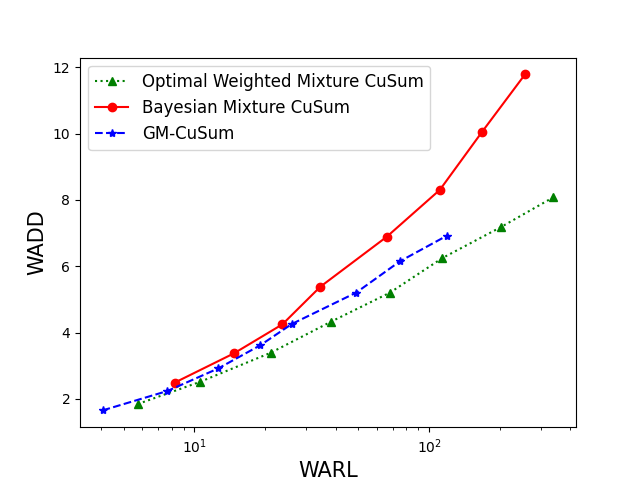}\par\caption{Comparison of the three algorithms in static setting: $n=4, K = 4$.}\label{fig:44}
	\end{multicols}
\end{figure*}

\begin{figure*}[!t]
	\begin{multicols}{4}
		\includegraphics[width=\linewidth]{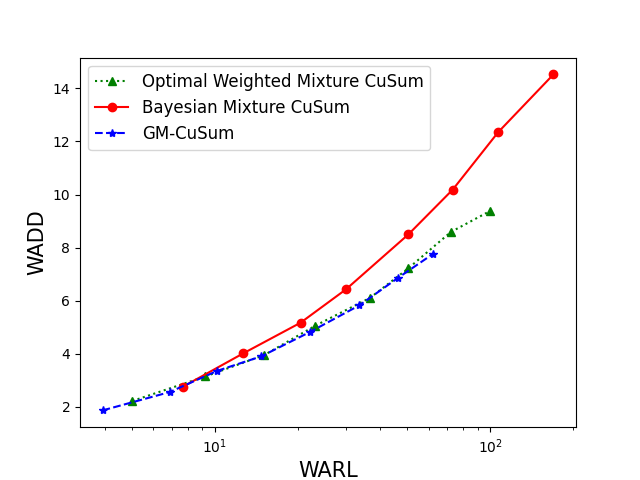}\par\caption{Comparison of the three algorithms in static setting: $n=8, K = 4$.}\label{fig:84}
		\includegraphics[width=\linewidth]{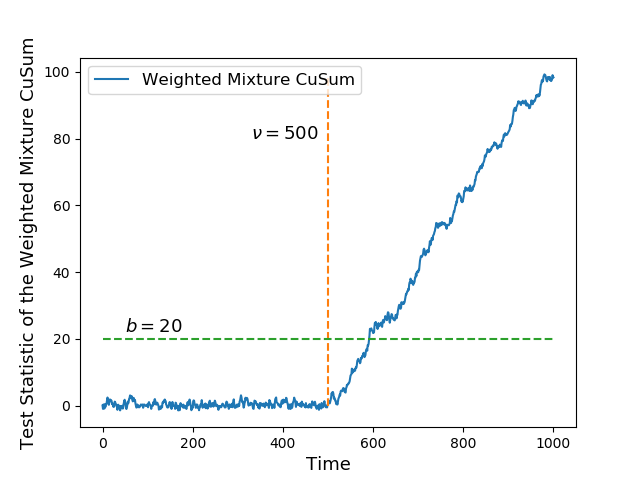}\par\caption{Evolution path of the weighted mixture CuSum algorithm.}\label{fig:3}
				\includegraphics[width=\linewidth]{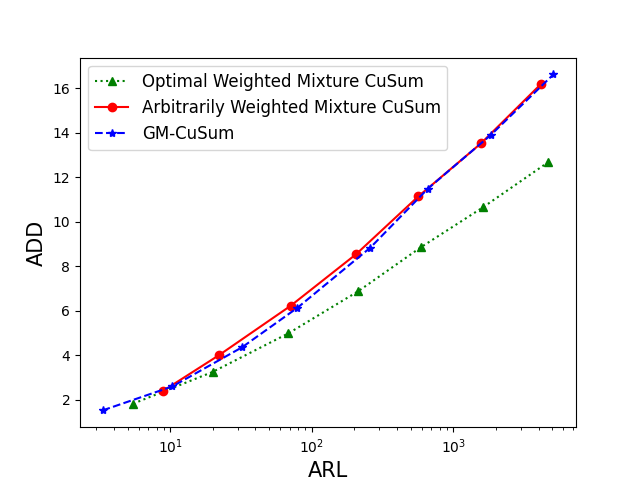}\par\caption{Comparison of the three algorithms in dynamic setting: $n=4, K = 2$.}\label{fig:d42}
		\includegraphics[width=\linewidth]{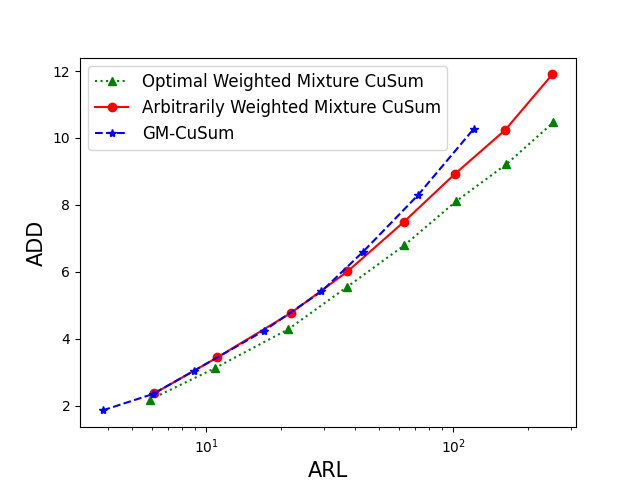}\par\caption{Comparison of the three algorithms in dynamic setting: $n=8, K = 2$.}\label{fig:d82}
	\end{multicols}
\end{figure*}

We then compare our GM-CuSum algorithm with a Bayesian mixture CuSum algorithm $T_B = \inf\Big\{t:\max_{1\leq j \leq t}\sum^t_{i = j}\log\frac{\frac{1}{|\mathcal{K}|}\sum_{k\in\mathcal{K}}\widetilde{\mP}^k(X^n[i])}{\widetilde{\mP}_0(X^n[i])}\geq b\Big\}$ and the optimal weighted mixture CuSum algorithm. We plot the WADD as a function of the WARL under the worst-case static trajectory. 

We consider four cases with different number of sensors and types. For the cases where there are two types of sensors,  for type \uppercase\expandafter{\romannumeral1} sensors, the pre- and post-change distributions are $\mathcal B(10,0.2)$ and $\mathcal B(10,0.5)$, for type \uppercase\expandafter{\romannumeral2} sensors, the pre- and post-change distributions are $\mathcal B(10,0.8)$ and $\mathcal B(10,0.6)$, respectively. We plot the figures for the cases where each type has two sensors and each type has four sensors in Fig.~\ref{fig:42} and Fig.~\ref{fig:82}, respectively. For the cases where there are four types of sensors, for type \uppercase\expandafter{\romannumeral1} sensors, the pre- and post-change distributions are $\mathcal B(10,0.2)$ and $\mathcal B(10,0.8)$, for type \uppercase\expandafter{\romannumeral2} sensors, the pre- and post-change distributions are $\mathcal B(10,0.3)$ and $\mathcal B(10,0.6)$, for type \uppercase\expandafter{\romannumeral3} sensors, the pre- and post-change distributions are $\mathcal B(10,0.5)$ and $\mathcal B(10,0.9)$, for type \uppercase\expandafter{\romannumeral4} sensors, the pre- and post-change distributions are $\mathcal B(10,0.4)$ and $\mathcal B(10,0.7)$ respectively. We plot the figures for the cases where each type has one sensor and each type has two sensors in Fig.~\ref{fig:44} and Fig.~\ref{fig:84}, respectively. We apply the Monte-Carlo approximation idea to obtain the optimal weight for our optimal weighted mixture CuSum algorithm. We repeat the experiment for 5000 times. 

It can be seen from Fig. \ref{fig:42}, Fig. \ref{fig:82}, Fig. \ref{fig:44} and Fig. \ref{fig:84} that our GM-CuSum outperforms the Bayesian algorithm and the performance of the optimal weighted mixture CuSum algorithm are close to the GM-CuSum algorithm. The simulation results show that our  optimal weighted mixture CuSum algorithm are also robust under the static setting. Moreover, the relationship between the WADD and log of the WARL is linear, which validates our theoretical results.

We then consider the dynamic anomaly. We use the same parameters of distributions as in the static setting. We first show an evolution path of the weighted mixture CuSum algorithm under a random trajectory ${\bm{S}}$ in Fig. \ref{fig:3}. To generate this random trajectory ${\bm{S}}$, at each time step, let the probability that one sensor of type \uppercase\expandafter{\romannumeral1} is affected be $0.8$ and the probability that one sensor of type \uppercase\expandafter{\romannumeral2} is affected be $0.2$. Similar to the GM-CuSum, before the change point, the test statistic fluctuates around zero, and after the change point, it starts to increase with a positive drift.

\begin{figure*}[!t]
	\begin{multicols}{3}
		\centering
				\includegraphics[width=0.7\linewidth]{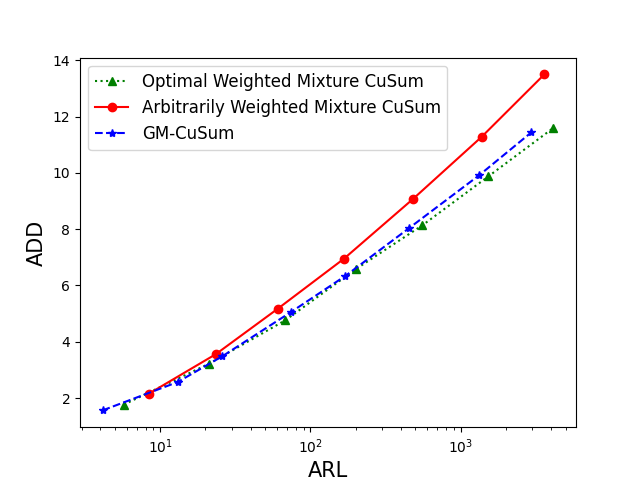}\par\caption{Comparison of the three algorithms in dynamic setting: $n=4, K = 4$.}\label{fig:d44}
		\includegraphics[width=0.7\linewidth]{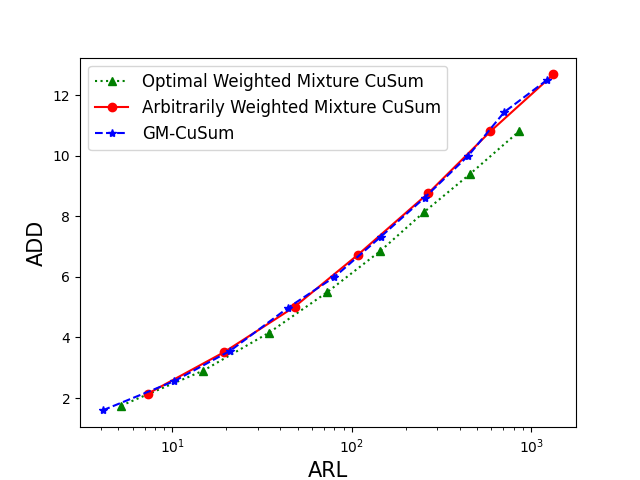}\par\caption{Comparison of the three algorithms in dynamic setting: $n=8, K = 4$.}\label{fig:d84}
		\includegraphics[width=0.7\linewidth]{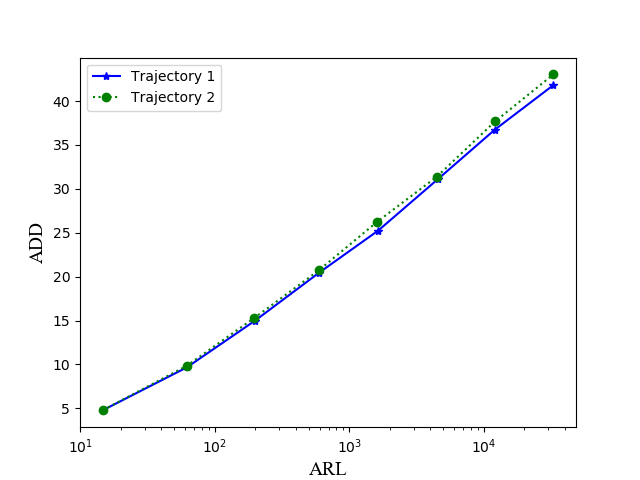}\par\caption{Optimal Weighted CuSum Algorithm under Different Trajectories.}
		\label{fig:5}
	\end{multicols}
\end{figure*}

We then compare our optimal weighted mixture CuSum algorithm with an arbitrary weighted mixture CuSum, i.e., replace $\bm{\beta^*}$ in \eqref{eq:wmixture} with some arbitrarily $\bm{\beta}$, e.g., $\bm{\beta} = (\frac{1}{2}, \frac{1}{2})$ for the case with two types 
and the GM-CuSum. Here, we plot the average detection delay (ADD) and the average run length (ARL) for some randomly generated trajectories.
It can be seen from Fig.~\ref{fig:d42}, Fig.~\ref{fig:d82}, Fig.~\ref{fig:d44} and Fig.~\ref{fig:d84} that our optimal weighted mixture CuSum algorithm outperforms the Bayesian weighted mixture CuSum algorithm and the GM-CuSum. The relationship between the WADD and log of the WARL is linear. It can also be observed that the GM-CuSum algorithm does not perform well under the dynamic setting.


We then compare the performance of our weighted mixture CuSum algorithm under two different trajectories. We choose $n = 2$ and $K = 2$. For type \uppercase\expandafter{\romannumeral1} sensors, the pre- and post-change distributions are $\mathcal B(10,0.3)$ and $\mathcal B(10,0.4)$, respectively. For type \uppercase\expandafter{\romannumeral2} sensors, the pre- and post-change distributions are $\mathcal B(10,0.8)$ and $\mathcal B(10,0.6)$, respectively. For trajectory 1, at each time, let the probability that one sensor of type one is affected be $0.8$ and the probability that one sensor of type two is affected be $0.2$. For trajectory 2, at each time, let the probability that one sensor of type one is affected be $0.2$ and the probability that one sensor of type two is affected be $0.8$. We plot the ADD as function of the ARL. It can be seen from Fig. \ref{fig:5} that for two different trajectories, our optimal weighted mixture CuSum algorithm have the same performance, which demonstrates the robustness of our optimal weighted mixture CuSum algorithm under different trajectories.

\section{conclusion}\label{sec:conclusion}
In this paper, we investigated the problem of quickest detection of an anomaly in  networks with unlabeled samples. We first investigated the case with a static anomaly. We used the MLE to estimate the type of the affected sensor. A GM-CuSum algorithm was proposed. We showed that it is second-order asymptotically optimal. We then extended our study to the case with a dynamic anomaly, that is, the affected sensor changes with time. We proposed a weighted mixture CuSum algorithm, and proved that it is first-order asymptotically optimal. Our approaches provide useful insights for general (sequential) statistical inference problems with unlabeled samples.

\appendices
\section{}\label{app:a}
Before the anomaly emerges, i.e., $t<\nu$, there are $n_k$ sensors in group $k$, $\forall 1\leq k\leq K$, and 0 sensors in group $k$, $\forall K<k\leq 2K$. Then, there are in total $\left(\substack{n\\n_1,\ldots,n_K}\right)$ possible $\sigma_t^{S[t]}$: $\{1,\ldots,n\}\rightarrow\{1,\ldots,K\}$ satisfying
$|\{i:\sigma_t^{S[t]}(i)=k\}|=n_k$, for any $k=1,\ldots,K$. We denote the collection of all such labels by $\mathcal S_{n,0}$.
After the anomaly emerges, i.e., $t\geq \nu$, one sensor of type $S[t]\neq 0$ is affected by anomaly. Therefore, the number of sensors in group $S[t]$ and $S[t]+K$ are $n_{S[t]}-1$ and 1 respectively. Then, there are $\left(\substack{n\\n_1,\ldots, n_{S[t]}-1, \ldots, n_K, 1}\right)$ possible $\sigma_t^{S[t]}$: $\{1,\ldots,n\}\rightarrow\{1,\ldots,K, S[t]+K\}$ satisfying
\begin{flalign*}
	|\{i:\sigma_t^{S[t]}(i)=k\}|=   \left\{\begin{array}{ll}
		n_k, &\text{ if } 1\leq k\leq K  \text{and}\ k\neq S[t],\\
		n_k-1, &\text{ if } k=S[t],\\
		1, &\text{ if } k=S[t]+K,\\
		0, &\text{ otherwise}.	
	\end{array}\right.
\end{flalign*} 
We then denote the collection of all such labels by $\mathcal S_{n,S[t]}$. 

Before the anomaly emerges, i.e., $t<\nu$, the samples $X^n[t]$ follows the distribution 
\begin{flalign}\label{eq:p0}
	\mP_{0,\sigma_t^{0}}(X^n[t])=\prod_{i=1}^np_{0,\sigma_t^{0}(i)}(X_i[t]),
\end{flalign}
for some unknown $\sigma_t^{0}\in\mathcal S_{n,0}.$  At time $t\geq\nu$, $X^n[t]$ follows the distribution  
\begin{flalign}\label{eq:p1}
	\mP^{S[t]}_{\sigma^{S[t]}_t}(X^n[t])\overset{\Delta}{=}&\prod\limits_{\substack{i:\sigma^{S[t]}_t(i) \leq K}}p_{0,\sigma^{S[t]}_t(i)}(X_i[t])\times \prod\limits_{\substack{i: \sigma_t^{S[t]}(i)> K}}p_{1,\sigma^{S[t]}_t(i)-K}(X_i[t]), 
\end{flalign}
for some unknown $\sigma^{S[t]}_t\in\mathcal S_{n,S[t]}$.

\section{Proof of Theorem \ref{theorem:1}}\label{sec:kaddlow}
Consider a simple QCD problem with a pre-change distribution  $\widetilde{\mP}_0$ and a post-change distribution  $\widetilde{\mP}^k$, respectively. Define the $\widetilde{\text{WADD}}_k$ and $\widetilde{\text{ARL}}$ for any stopping rule $\tau$ as follows:
\begin{flalign}
\widetilde{\text{WADD}}_k(\tau)&=\sup\limits_{\nu\geq1}\text{esssup}\widetilde{\mE}^{k,\nu}[(\tau-\nu)^+|\widetilde{\mathbf X}^n[1,\nu-1]],\nn\\
\widetilde{\text{ARL}}(\tau)&=\widetilde{\mE}^\infty[\tau],
\end{flalign}
where $\widetilde{\mE}^{k,\nu}$ denotes the expectation when the change is at $\nu$, the pre- and post-change distributions are $\widetilde{\mP}_0$ and $\widetilde{\mP}^k$, and $\widetilde{\mathbf X}^n[t]$ for $1 \leq t\leq\nu-1$ are i.i.d.\ from $\widetilde{\mP}_0$, $\widetilde{\mE}^\infty$ denotes the expectation when there is no change and samples are generated according to $\widetilde{\mP}_0$.

For any $1\leq k\leq K$, consider another QCD problem with a pre-change distribution  $\mP_{0,\sigma_t^0}$ and a post-change distribution  $\mP^{k}_{\sigma_t^k}$, respectively. For this pair of pre- and post-change distributions, define the ${\text{WADD}}_k$ and ${\text{WARL}}$ for any stopping rule $\tau$ as follows:
\begin{flalign}\label{eq:fixk}
\text{WADD}_k(\tau) &= \sup_{\nu\geq 1}\sup\limits_{\Omega_k}\text{esssup}\mE^{k,\nu}_{\Omega_k}[(\tau-\nu)^+|{\mathbf X}^n[1,\nu-1]], \nn\\ \text{WARL}(\tau) &=\inf\limits_{\Omega} {\mE}^\infty_{\Omega}[\tau].
\end{flalign}

For any $1\leq k\leq K$ and any $\tau$ satisfying $\text{WARL}(\tau)\geq \gamma$, it can be shown that
\begin{flalign}
\text{WADD}(\tau)&= \sup_{k\in\mathcal{K}}\text{WADD}_k(\tau) \nn\\&\geq\sup_{\nu \geq 1}\sup_{\Omega_k}\text{esssup}\mE^{k,\nu}_{\Omega_k}[(\tau-\nu)^+|\mathbf X^n[1,\nu-1]]\nn\\&\geq \sup\limits_{\nu\geq1}\text{esssup}\widetilde{\mE}^{k,\nu}[(\tau-\nu)^+|\widetilde{\mathbf X}^n[1,\nu-1]]\nn\\& = \widetilde{\text{WADD}}_k(\tau).
\end{flalign}
The second inequality is due to the fact that for any $\tau$, $\text{WADD}_k(\tau)\geq \widetilde{\text{WADD}}_k(\tau)$ \cite[eq. (18)]{sun2020tspanonymous}. Similarly, we have that for any $\tau$, $\text{WARL}(\tau)\leq \widetilde{\text{ARL}}(\tau)$ \cite[eq. (18)]{sun2020tspanonymous}. It then follows that for any $k\in\mathcal{K}$,
\begin{flalign}
\inf_{\tau:\text{WARL}(\tau)\geq \gamma} \text{WADD}(\tau) &\geq \inf_{\tau:\widetilde{\text{ARL}}(\tau)\geq \gamma} \widetilde{\text{WADD}}_k(\tau)\nn\\&\geq\frac{\log\gamma}{I_k}+O(1),\ \text{as}\ \gamma \rightarrow \infty.
\end{flalign}
The last inequality is due to the universal lower bound on WADD for a simple QCD problem\cite{lai1998information}.
We then have that 
\begin{flalign}\label{eq:kmixlow}
\inf_{\tau:\text{WARL}(\tau)\geq \gamma} \text{WADD}(\tau) \geq \frac{\log\gamma}{I^*}+O(1),\ \text{as}\ \gamma\rightarrow\infty.
\end{flalign}

\section{Proof of Theorem \ref{THEOREM:2}}\label{sec:karlproof}
For any $m\geq 0$, let $r_0 = 0$ and define the stopping time 
\begin{flalign}
r_{m+1} = \inf\Big\{t>r_m: \sup\limits_k\sum\limits^t_{i=r_m+1}\log\frac{\widetilde{\mP}^k(X^n_i)}{\widetilde{\mP}_0(X^n_i)}\leq 0\Big\}.
\end{flalign}
 
For any permutation $\pi(X^n)=(X_{\pi(1)},X_{\pi(2)},\ldots,X_{\pi(n)})$, we have that  $\log\frac{\widetilde{\mP}^k(X^n)}{\widetilde{\mP}_0(X^n)}=\log\frac{\widetilde{\mP}^k(\pi(X^n))}{\widetilde{\mP}_0(\pi(X^n))}$. For any $\pi$, let $\hat{\sigma}^0=\sigma^0\circ\pi$, where ``$\circ$'' denotes the composition of two functions. Then $\mE_{0,\sigma^0}\Big[\log\frac{\widetilde{\mP}^k(\pi(X^n))}{\widetilde{\mP}_0(\pi(X^n))}\Big]=\mE_{0,\sigma^0\circ\pi}\Big[\log\frac{\widetilde{\mP}^k(X^n)}{\widetilde{\mP}_0(X^n)}\Big]=\mE_{0,\hat{\sigma}^0}\Big[\log\frac{\widetilde{\mP}^k(X^n)}{\widetilde{\mP}_0(X^n)}\Big]$. For any $\hat{\sigma}^0\in\msn$, a $\pi$ can always be found so that $\sigma^0\circ\pi=\hat{\sigma}^0$. Thus, for any $\sigma^0,\hat{\sigma}^0\in\msn$,
$
\mE_{0,\hat{\sigma}^0}\Big[\log\frac{\widetilde{\mP}^k(X^n)}{\widetilde{\mP}_0(X^n)}\Big]=\mE_{0,\sigma^0}\Big[\log\frac{\widetilde{\mP}^k(X^n)}{\widetilde{\mP}_0(X^n)}\Big].
$

We then have that for any $\sigma^0\in\msn$,
\begin{flalign}
&\mE_{0,\sigma^0}\Big[\frac{\widetilde{\mP}^k(X^n)}{\widetilde{\mP}_0(X^n)}\Big]\nn\\ &= \frac{1}{\mid\msn\mid}\sum\limits_{\sigma^0\in \msn}\mE_{0,\sigma^0}\Big[\frac{\widetilde{\mP}^k(X^n)}{\widetilde{\mP}_0(X^n)}\Big] \nn\\& = \frac{1}{\mid\msn\mid}\sum\limits_{\sigma^0\in \msn}\int \frac{\widetilde{\mP}^k(x^n)}{\widetilde{\mP}_0(x^n)}\cdot\mP_{0,\sigma^0}(x^n)\mathrm{d}x^n \nn\\& = \int \frac{\widetilde{\mP}^k(x^n)}{\widetilde{\mP}_0(x^n)}\cdot\widetilde{\mP}_0(x^n)\mathrm{d}x^n \nn\\& = \int \widetilde{\mP}^k(x^n)\mathrm{d}x^n = 1.
\end{flalign}
Therefore, for any $\Omega$ and $t>r_m$, 
\begin{flalign}
&\mE_{\Omega}^\infty\Bigg[\prod\limits_{i = r_m+1}^{t+1}\frac{\widetilde{\mP}^k(X^n_i)}{\widetilde{\mP}_0(X^n_i)}\Bigg|\mathcal{F}_t\Bigg]\nn\\&= \mE_{\Omega}^\infty\Bigg[\prod\limits_{i = r_m+1}^{t}\frac{\widetilde{\mP}^k(X^n_i)}{\widetilde{\mP}_0(X^n_i)}\cdot \frac{\widetilde{\mP}^k(X^n_{t+1})}{\widetilde{\mP}_0(X^n_{t+1})}\Bigg|\mathcal{F}_t\Bigg]\nn\\&= \mE_{\Omega}^\infty\Bigg[\prod\limits_{i = r_m+1}^{t}\frac{\widetilde{\mP}^k(X^n_i)}{\widetilde{\mP}_0(X^n_i)}\Bigg|\mathcal{F}_t\Bigg]\cdot \mE_{0,\sigma^0}\Bigg[ \frac{\widetilde{\mP}^k(X^n_{t+1})}{\widetilde{\mP}_0(X^n_{t+1})}\Bigg|\mathcal{F}_t\Bigg]\nn\\&= \prod\limits_{i = r_m+1}^{t}\frac{\widetilde{\mP}^k(X^n_i)}{\widetilde{\mP}_0(X^n_i)}\cdot\mE_{0,\sigma^0}\Bigg[\frac{\widetilde{\mP}^k(X^n_{t+1})}{\widetilde{\mP}_0(X^n_{t+1})}\Bigg]\nn\\&=\prod\limits_{i = r_m+1}^{t}\frac{\widetilde{\mP}^k(X^n_i)}{\widetilde{\mP}_0(X^n_i)}.
\end{flalign}
Therefore, $\Big\{\prod\limits_{i = r_m+1}^t\frac{\widetilde{\mP}^k(X^n_i)}{\widetilde{\mP}_0(X^n_i)},\mathcal{F}_t, t> r_m\Big\}$ is a martingale under $\mP_{\Omega}^\infty$ for any $\Omega$ with mean 1. 

We then have that for any $\Omega$,
\begin{flalign}
&\mP_\Omega^\infty\Bigg\{\sup\limits_k\sum\limits^t_{i=r_m+1}\log\frac{\widetilde{\mP}^k(X^n_i)}{\widetilde{\mP}_0(X^n_i)}\geq b\ \text{for some}\  t>r_m\Bigg|\mathcal{F}_{r_m}\Bigg\}\nn\\&\leq \sum\limits_{k=1}^K\mP_\Omega^\infty\Bigg\{\sum\limits^t_{i=r_m+1}\log\frac{\widetilde{\mP}^k(X^n_i)}{\widetilde{\mP}_0(X^n_i)}\geq b\ \text{for some}\  t>r_m\Bigg|\mathcal{F}_{r_m}\Bigg\}\nn\\&= \sum\limits_{k=1}^K\mP_\Omega^\infty\Bigg\{\prod\limits^t_{i=r_m+1}\frac{\widetilde{\mP}^k(X^n_i)}{\widetilde{\mP}_0(X^n_i)}\geq e^b\ \text{for some}\  t>r_m\Bigg|\mathcal{F}_{r_m}\Bigg\}\nn\\& \leq K\frac{\mE_{0,\sigma^0}\Big[\frac{\widetilde{\mP}^k(X^n_{r_m+1})}{\widetilde{\mP}_0(X^n_{r_m+1})}\Big]}{e^b} =Ke^{-b},
\end{flalign}
where the last inequality is due to Doob's submartingale inequality\cite{williams_1991} and the optional sampling theorem\cite{williams_1991}.

Let $M = \inf\Big\{m\geq 0: r_m<\infty \ \text{and} \  \sup\limits_k\sum\limits^t_{i=r_m+1}\log\frac{\widetilde{\mP}^k(X^n_i)}{\widetilde{\mP}_0(X^n_i)}\geq b \ \text{for some}\  t>r_m \Big\}$. We have that for any $\Omega$,
\begin{flalign}
&\mP_\Omega^\infty\big(M\geq m+1|\mathcal{F}_{r_m}\big)\nn\\& \geq \mP_\Omega^\infty\Bigg\{\sup\limits_k\sum\limits^t_{i=r_m+1}\log\frac{\widetilde{\mP}^k(X^n_i)}{\widetilde{\mP}_0(X^n_i)}< b\ \text{for all}\  t>r_m\Bigg|\mathcal{F}_{r_m}\Bigg\}\nn\\&\geq 1-Ke^{-b}.
\end{flalign}
We then have that for any $\Omega$,
\begin{flalign}
\mP_\Omega^\infty(M>m) &= \mE_\Omega^\infty\Big[\mP_\Omega^\infty\big(M\geq m+1|\mathcal{F}_{r_m}\big)\cdot\mathbbm{1}_{\{M\geq m\}}\Big]\nn\\&\geq (1-Ke^{-b})\mP_\Omega^\infty(M>m-1)\nn\\&\geq(1-Ke^{-b})^2\mP_\Omega^\infty(M>m-2) \nn\\&\geq(1-Ke^{-b})^m\mP_\Omega^\infty(M>0)\nn\\&=(1-Ke^{-b})^m.
\end{flalign}
It then follows that
\begin{flalign}
\text{WARL}(T_G)& = \inf\limits_{\Omega}\mE_\Omega^\infty[T_G]\geq \inf\limits_{\Omega} \mE_\Omega^\infty[M] \nn\\&\geq\inf\limits_{\Omega} \sum_{m=0}^{\infty}\mP_\Omega^\infty(M>m)\nn\\&\geq \sum_{m=0}^{\infty}(1-Ke^{-b})^m = \frac{e^b}{K}.
\end{flalign}
Let $b = \log K\gamma$, we have that $\text{WARL}(T_G) \geq \gamma$.
Let $T_k$ be the mixture CuSum algorithm for problem in \eqref{eq:fixk}:
\begin{flalign}
T_k = \inf\bigg\{t:\max\limits_{1\leq j \leq t}\sum\limits^t_{i = j}\log\frac{\widetilde{\mP}^k(X^n[i])}{\widetilde{\mP}_0(X^n[i])}\geq b\bigg\}.
\end{flalign}
It then follows that for any $1\leq k\leq K$,
\begin{flalign}
\text{WADD}_k(T_G) &= \sup_{\nu \geq 1}\sup_{\Omega_k}\text{esssup}\mE^{k,\nu}_{\Omega_k}[(T_G-\nu)^+|\mathbf X^n[1,\nu-1]]\nn\\&\leq \sup_{\nu \geq 1}\sup_{\Omega_k}\text{esssup}\mE^{k,\nu}_{\Omega_k}[(T_k-\nu)^+|\mathbf X^n[1,\nu-1]]\nn\\&\leq \frac{\log b}{I_k} + O(1),
\end{flalign}
where the last equality is because of the exact optimality of the mixture CuSum algorithm (see Theorem 1 in \cite{sun2020tspanonymous}).

To satisfy the WARL constraint, choose $b = \log {K\gamma}$, we then have that
\begin{flalign}\label{eq:kwadd}
\text{WADD}(T_G) & = \sup\limits_{k\in\mathcal K}\text{WADD}_k(T_G)\leq \sup\limits_{k\in\mathcal K}\text{WADD}_k(T_k)\nn\\& = \sup\limits_{k\in\mathcal K}\frac{\log K\gamma}{I_k} + O(1) \nn\\&= \frac{\log\gamma}{I^*} + \frac{\log K}{I^*} + O(1), \ \text{as}\  \gamma \rightarrow \infty.
\end{flalign}

\section{Proof of Theorem \ref{theorem:3}}\label{sec:lowwadds}
For any trajectory ${\bm{S}}$ and stopping time $\tau$, define the WADD and WARL
\begin{flalign}\label{eq:worstP}
	\text{WADD}_{\bm{S}}(\tau)&= \sup_{\nu\geq 1} \sup_{\Omega_{\bm{S}}}\text{esssup}\mE^{{\bm{S}},\nu}_{\Omega_{\bm{S}}}[(\tau-\nu)^+|\mathbf X^n[1,\nu-1]],\nn\\\text{ARL}(\tau) &= \inf_{\Omega}\mE_{\Omega}^\infty[\tau].
\end{flalign}
Consider QCD problem with a pre-change distribution  $\widetilde{\mP}_0=\frac{1}{\mid \msn\mid}\sum_{\sigma^0\in \msn}\mP_{0,\sigma^0}$ and a post-change distribution  $\widetilde{\mP}^{S[t]}=\frac{1}{\mid \mss\mid}\sum_{\sigma^{S[t]}\in \mss}\mP^{S[t]}_{\sigma^{S[t]}}$, respectively. For this pair of pre- and post-change distributions and any trajectory ${\bm{S}}$, define the $\widetilde{\text{WADD}}_{\bm{S}}$ and $\widetilde{\text{ARL}}_{\bm{S}}$ for any stopping rule $\tau$:
\begin{flalign}\label{eq:fixS}
	\widetilde{\text{WADD}}_{\bm{S}}(\tau) &= \sup_{\nu\geq 1}\text{esssup}\widetilde{\mE}^{\bm{S},\nu}[(\tau-\nu)^+|\widetilde{\mathbf X}^n[1,\nu-1]], \nn\\ \widetilde{\text{ARL}}(\tau) &= \widetilde{\mE}^\infty[\tau].
\end{flalign}
where $\widetilde{\mE}^{\bm{S},\nu}$ denotes the expectation when change point is $\nu$, before the change point, the data follows distribution $\widetilde{\mP}_0$ and after the change point, at time $t$, the data follows the distribution $\widetilde{\mP}^{S[t]}$, and $\widetilde{\mathbf X}^n[1,\nu-1]$ are i.i.d.\ from $\widetilde{\mP}_0$; and $\widetilde{\mE}^\infty$ denote the expectation when for any $t\geq 0$, the data follows distribution $\widetilde{\mP}_0$, i.e., $\nu=\infty$.

Consider another QCD problem with pre-change distribution $\widetilde{\mP}_0$ and post-change distribution $\widetilde{\mP}^{\bm{\beta^*}}$. Under this pair of pre- and post-change distributions, for any stopping time $\tau$, define worst-case average detection delay and average running length as follows:
\begin{flalign}\label{eq:fixb}
	\widetilde{\text{WADD}}_{\bm{\beta^*}}(\tau) &= \sup_{\nu\geq1}\text{esssup}\widetilde{\mE}^{\bm{\beta^*},\nu}[(\tau-\nu)^+|\widetilde{\mathbf X}^n[1,\nu-1]],\nn\\ \widetilde{\text{ARL}}(\tau) &= \widetilde{\mE}^\infty[\tau].
\end{flalign}
In QCD problems, ARL only depends on the pre-change distribution. Therefore, for any stopping time $\tau$, problems in \eqref{eq:warl} and \eqref{eq:worstP} have the same ARL, problems in \eqref{eq:fixS} and \eqref{eq:fixb} have the same ARL. Let $\mathcal{C}_\gamma$ denotes the collection of all stopping times $\tau$ that satisfy $\text{ARL}(\tau)\geq \gamma$ and $\widetilde{\mathcal{C}}_\gamma$ denotes the collection of all stopping times $\tau$ that satisfy $\widetilde{\text{ARL}}(\tau)\geq\gamma$.
Our goal is to prove that
\begin{flalign}\label{eq:dytheorem}
	\hspace{-0.26cm}\inf_{\tau\in \mathcal{C}_\gamma}\text{WADD}(\tau) \geq \inf_{\tau\in \widetilde{\mathcal{C}}_\gamma}\widetilde{\text{WADD}}_{\bm{\beta^*}}(\tau) \sim \frac{\log\gamma}{I_{\bm{\beta^*}}}(1+o(1)).
\end{flalign}

Construct a new sequence of random variables $\{\widehat{X}^n[t]\}_{t=1}^\infty$. Before the change point, $\widehat{X}^n[t]$ are i.i.d. according to the mixture distribution $\widetilde{\mP}_0=\frac{1}{\mid \msn\mid}\sum_{\sigma^0\in \msn}\mP_{0,\sigma^0}$. After the change point, i.e., $t\geq \nu$, $\widehat{X}^n[t]$ follows the distribution $\mP^{S[t]}_{\sigma^{S[t]}_t}$ for some $\sigma^{S[t]}_t\in\mss$. Specifically,
\begin{flalign}\label{eq:new}
	\widehat{X}^n[t] \sim \left\{\begin{array}{ll}
		\widetilde{\mP}_0, &\text{ if } t<\nu, \\
		\mP^{S[t]}_{\sigma^{S[t]}_t}, &\text{ if } t\geq \nu.
	\end{array}\right.
\end{flalign}

For any stopping time $\tau$ and any ${\bm{S}}$, define the worst-case average detection delay for the model in \eqref{eq:new} as follows:
\begin{flalign}
	\widehat{\text{WADD}_{\bm{S}}}(\tau)=&\sup\limits_{\nu\geq1}\sup_{\sigma_\nu^{S[\nu]},...,\sigma_\infty^{S[\infty]}}\text{esssup}{\widehat{\mE}}^{{\bm{S}},\nu}_{\sigma_\nu^{S[\nu]},...,\sigma_\infty^{S[\infty]}}[(\tau-\nu)^+|\widehat{\mathbf X}^n[1,\nu-1]],
\end{flalign}
where $\widehat{\mE}^{{\bm{S}},\nu}_{\sigma_\nu^{S[\nu]},...,\sigma_\infty^{S[\infty]}}$ denotes the expectation when the data is distributed according to \eqref{eq:new}.

Let $\widetilde{\text{WADD}}(\tau) = \sup\limits_{{\bm{S}}}\widetilde{\text{WADD}}_{\bm{S}}(\tau)$. To prove \eqref{eq:dytheorem}, we will first show that for any ${\bm{S}}$,  ${\text{WADD}_{\bm{S}}}(\tau) =\widehat{\text{WADD}_{\bm{S}}}(\tau)$, and then show that $\widehat{\text{WADD}_{\bm{S}}}(\tau) \geq \widetilde{\text{WADD}_{\bm{S}}}(\tau)$. We will then complete our proof by showing that for any $\tau$ and $\bm{\beta}$, $\widetilde{\text{WADD}}(\tau) \geq \widetilde{\text{WADD}}_{\bm{\beta}}(\tau)$.

\textbf{Step 1.} Denote by $\mathcal{M}$ the collection of all $\{\sigma_1^0,...,\sigma_{\nu-1}^0\}$, and $\mu$ is an element in $\mathcal{M}$. When the trajectory is ${\bm{S}}$, denote by $\mathcal{N}_{\bm{S}}$ the collection of all $\{\sigma_\nu^{S[\nu]},...,\sigma_\infty^{S[\infty]}\}$, and $\omega$ is an element in $\mathcal{N}_{\bm{S}}$. Then, the $\text{WADD}_{\bm{S}}$ can be written as
\begin{flalign*}
	\text{WADD}_{\bm{S}}(\tau)&= \sup\limits_{\nu\geq1}\sup_{\Omega_{\bm{S}}}\text{esssup}{\mE}^{{\bm{S}},\nu}_{\Omega_{\bm{S}}}[(\tau-\nu)^+|{\mathbf X}^n[1,\nu-1]]\nn \\&= \sup\limits_{\nu\geq1}\sup_{\omega\in\mathcal{N}_{\bm{S}}}\sup\limits_{\mu\in\mathcal{M}}\text{esssup}{\mE}^{{\bm{S}},\nu}_{\omega}[(\tau-\nu)^+|{\mathbf X}^n[1,\nu-1]],
\end{flalign*}
where $\mE^{{\bm{S}},\nu}_\omega$ denotes the expectation when change point is $\nu$, the trajectory is ${\bm{S}}$, and after the change point, the data follows distribution $\prod_{t=\nu}^\infty\mP^{S[t]}_{\sigma^{S[t]}_t}$.
We note that $\widehat{X}^n[t]$ and $X^n[t]$, for $t\geq\nu$, have the same distribution $\mP^{S[t]}_{\sigma^{S[t]}_t}$. Therefore, the difference between $\text{WADD}_{\bm{S}}$ and $\widehat{\text{WADD}_{\bm{S}}}$ lies in that they take esssup with respect to different distributions, i.e., the distributions of $\mathbf{X}^n[1,\nu-1]$ and $\mathbf{\widehat{X}}^n[1,\nu-1]$ are different. Let $f_{\omega}(\mathbf {X}^n[1,\nu-1])$ denote ${\mE}^{{\bm{S}},\nu}_{\omega}[(\tau-\nu)^+|\mathbf{X}^n[1,\nu-1]]$.
Then, $\text{WADD}_{\bm{S}}$ and $\widehat{\text{WADD}_{\bm{S}}}$ can be written as
\begin{flalign}
	\text{WADD}_{\bm{S}}(\tau) &= \sup\limits_{\nu\geq1}\sup_{\omega\in\mathcal{N}_{\bm{S}}}\sup_{\mu\in\mathcal{M}}\text{esssup} f_{\omega}({\mathbf X}^n[1,\nu-1]),\nn\\\widehat{\text{WADD}_{\bm{S}}}(\tau) &= \sup\limits_{\nu\geq1}\sup_{\omega\in\mathcal{N}_{\bm{S}}}\text{esssup} f_{\omega}(\widehat{\mathbf X}^n[1,\nu-1]).
\end{flalign}
It then suffices to show that for any $\omega \in \mathcal{N}_S$,
$
	\sup_{\mu\in\mathcal{M}}\text{esssup} f_{\omega}({\mathbf X}^n[1,\nu-1])=\text{ess}\sup f_{\omega}(\widehat{\mathbf X}^n[1,\nu-1]).\nn 
$

For any $\omega \in \mathcal{N}_{\bm{S}}$ and $\mu \in \mathcal{M}$, let
\begin{flalign}
	b_{\omega,\mu} &= \text{esssup} f_\omega({\mathbf X}^n[1,\nu-1])\nn\\&=\inf\{b: \mP_{\mu}(f_{\omega}({\mathbf X}^n[1,\nu-1])>b)=0\},
\end{flalign}
where $\mP_\mu$ denotes the probability measure when the data is generated from $\mP_{0,\sigma_1^0},...,\mP_{0,\sigma_{\nu-1}^0}$ before change point $\nu$.

Let $b^*_\omega = \text{esssup} f_\omega(\widehat{\mathbf{X}}^n[1,\nu-1])$. It can be shown that
\begin{flalign}
	b^*_\omega &=
	\inf\bigg\{b:\int_{\textbf{x}^n[1,\nu-1]} \mathbbm{1}_{\{f_\omega(\textbf{x}^n[1,\nu-1])>b\}}\times\mathrm{d} \prod_{t=1}^{\nu-1}\widetilde{\mP}_0(x^n(t))=0\bigg\}
	\nn\\&=\inf\bigg\{b:\int_{\textbf{x}^n[1,\nu-1]} \mathbbm{1}_{\{f_\omega(\textbf{x}^n[1,\nu-1])>b\}}\times\mathrm{d} \prod_{t=1}^{\nu-1}\frac{1}{\mid \msn \mid}\sum_{\sigma_t^0 \in \msn}{\mP}_{0,\sigma_t^0}(x^n(t))=0\bigg\}
	\nn\\&=\inf\bigg\{b:\int_{\textbf{x}^n[1,\nu-1]} \mathbbm{1}_{\{f_\omega(\textbf{x}^n[1,\nu-1])>b\}}\times\mathrm{d}\frac{1}{\mid \mathcal{M}\mid}\sum_{\mu\in \mathcal{M}}{\mP}_{\mu}(\textbf{x}^n[1,\nu-1])=0\bigg\}
	\nn\\&=\inf\bigg\{b: \frac{1}{\mid \mathcal{M}\mid}\sum_{\mu\in \mathcal{M}}{\mP}_{\mu}(f_{\omega}({{\mathbf X}}^n[1,\nu-1])>b)=0\bigg\}.\nn
\end{flalign}
It then follows that for any $\mu\in\mathcal{M}$, and $\omega \in \mathcal{N}_S$,
$
	{\mP}_{\mu}(f_{\omega}({{\mathbf X}}^n[1,\nu-1])>b^*_\omega) = 0.
$
Therefore, for any $\mu\in\mathcal{M}$, we have that $b_{\omega,\mu}\leq b^*_\omega$. Then
\begin{flalign}\label{eq:achieve}
	\sup\limits_{\mu\in\mathcal{M}} b_{\omega,\mu} \leq b^*_\omega.
\end{flalign}

Conversely, for any $\mu\in\mathcal{M}$, we have that
$
	\mP_{\mu} \Big(f_{\omega}({\mathbf X}^n[1,\nu-1])>\sup\limits_{\mu\in\mathcal{M}} b_{\omega,\mu}\Big)= 0.
$
Then,
$
	\frac{1}{\mid \mathcal{M}\mid}\sum_{\mu \in \mathcal{M}}\mP_{\mu}\Big(f_{\omega}({\mathbf X}^n[1,\nu-1])>\sup\limits_{\mu\in\mathcal{M}} b_{\omega,\mu}\Big) = 0.
$
This further implies that 
\begin{flalign}\label{eq:converse}
	b^*_\omega \leq \sup\limits_{\mu\in\mathcal{M}} b_{\omega,\mu}.
\end{flalign}

Combining \eqref{eq:achieve} and \eqref{eq:converse}, we have that $\sup_{\mu\in\mathcal{M}} b_{\omega,\mu}=b^*_\omega,$
and thus
$
	\sup_{\mu\in\mathcal{M}}\text{esssup} f_{\omega}({\mathbf X}^n[1,\nu-1])= \text{esssup} f_{\omega}(\widehat{\mathbf X}^n[1,\nu-1]).\nn
$
This implies that for any $\tau$,
\begin{flalign}\label{eq:first}
	\text{WADD}_{\bm{S}}(\tau) = \widehat{\text{WADD}_{\bm{S}}}(\tau).
\end{flalign}

\textbf{Step 2.} The next step is to show that $\widehat{\text{WADD}}_{\bm{S}}(\tau) \geq \widetilde{\text{WADD}}_{\bm{S}}(\tau)$. We will first show that
$
	\sup_{\omega\in\mathcal{N}_{\bm{S}}}{\text{esssup} f_{\omega}(\widehat{\mathbf X}^n[1,\nu-1])}\geq {\text{esssup}\sup_{\omega\in\mathcal{N}_{\bm{S}}} f_{\omega}(\widehat{\mathbf X}^n[1,\nu-1])}.
$
Denote by $\widetilde{\mP}^\nu$ the probability measure when the change is at $\nu$, the pre- and post-change distributions are $\widetilde{\mP}_0$ and $\widetilde{\mP}^{S[t]}$ at time $t$, respectively. Let $\hat{b}=\sup_{\omega\in\mathcal{N}_{\bm{S}}}{\text{esssup} f_{\omega}(\widehat{\mathbf X}^n[1,\nu-1])}$.
For any $\omega\in\mathcal{N}_{\bm{S}}$, we have that
$
	\widetilde{\mP}^\nu\Big(f_{\omega}(\widehat{\mathbf X}^n[1,\nu-1])> \hat{b}\Big) = 0.
$
Since $\mathcal{N}_{\bm{S}}$ is countable, and a countable union of sets of measure zero has measure zero, we then have that 
\begin{flalign}
	&\widetilde{\mP}^\nu\Big(\sup\limits_{\omega\in\mathcal{N}_{\bm{S}}} f_{\omega}(\widehat{\mathbf X}^n[1,\nu-1])> \hat{b}\Big)\nn\\& \leq \widetilde{\mP}^\nu\Big(\cup_{\omega\in\mathcal{N}_{\bm{S}}}\big\{f_\omega(\widehat{\mathbf X}^n[1,\nu-1])>\hat{b}\big\}\Big) = 0.
\end{flalign}
Therefore,
\begin{flalign}\label{eq:supp}
	\hat{b}&=\sup\limits_{\omega\in\mathcal{N}_{\bm{S}}} \text{esssup} f_{\omega}(\widehat{\mathbf X}^n[1,\nu-1]) \nn\\&\geq \text{esssup}\sup\limits_{\omega\in\mathcal{N}_{\bm{S}}} f_{\omega}(\widehat{\mathbf X}^n[1,\nu-1]).
\end{flalign}

Before the change point $\nu$, $\widehat{X}^n[t]$ and $\widetilde{X}^n[t]$ follow the same distribution. For any $T\geq\nu+1$, we have that
\begin{flalign}
	&\sup_{\substack{\{\sigma_\nu^{S[\nu]},\cdots,\sigma_T^{S[T]}\}\\\in\mathcal S_{n,S[\nu]} \times,\cdots,\times \mathcal S_{n,S[T]}}}\sum_{t=\nu+1}^T (t-\nu)\mP^{{\bm{S}},\nu}_{\sigma_\nu^{S[\nu]},\cdots,\sigma_T^{S[T]}}(\tau=t|\widehat{\mathbf{X}}^n[1,\nu-1])\nn\\
	&\geq \sum_{t=\nu+1}^T (t-\nu)\frac{1}{\mid \mathcal S_{n,S[\nu]}\mid \times\cdots \times \mid\mathcal S_{n,S[T]} \mid}\sum\limits_{\substack{\{\sigma_\nu^{S[\nu]},\cdots,\sigma_T^{S[T]}\}\\\in\mathcal S_{n,S[\nu]} \times,\cdots,\times \mathcal S_{n,S[T]}}}\mP^{{\bm{S}},\nu}_{\sigma_\nu^{S[\nu]},\cdots,\sigma_T^{S[T]}}(\tau=t|\widehat{\mathbf{X}}^n[1,\nu-1])\nn\\&=\sum_{t=\nu+1}^T (t-\nu)\widetilde{\mP}^{\bm{S},\nu}(\tau=t|\widetilde{\mathbf{X}}^n[1,\nu-1]),
\end{flalign}
where $\mP_{\sigma_\nu^{S[\nu]},...,\sigma_T^{S[T]}}^{{\bm{S}},\nu}$ denotes the probability measure when change point is $\nu$, the trajectory is ${\bm{S}}$, the observations from time $\nu$ to time $T$ are generated according to ${\mP}^{S[\nu]}_{\sigma_\nu^{S[\nu]}},...,{\mP}^{S[T]}_{\sigma_T^{S[T]}}$.
As $T\rightarrow\infty$, we have that
\begin{flalign}\label{eq:ave}
	f_{\omega}(\widehat{\mathbf X}^n[1,\nu-1])\geq
	\widetilde{\mE}^{\bm{S},\nu}[(\tau-\nu)^+|\widetilde{\mathbf X}^n[1,\nu-1]],
\end{flalign}

From \eqref{eq:supp} and \eqref{eq:ave}, we have that
\begin{flalign}\label{eq:second}
	\widehat{\text{WADD}_{\bm{S}}}(\tau) &= \sup\limits_{\omega\in\mathcal{N}_{\bm{S}}} \text{esssup} f_{\omega}(\widehat{\mathbf X}^n[1,\nu-1])\nn\\&\geq \text{esssup}\widetilde{\mE}^{\bm{S},\nu}[(\tau-\nu)^+|\widetilde{\mathbf X}^n[1,\nu-1]]\nn\\&=
	\widetilde{\text{WADD}}_{\bm{S}}(\tau).
\end{flalign}

Combining \eqref{eq:first} and \eqref{eq:second}, it follows that
\begin{flalign}
	\text{WADD}_{\bm{S}}(\tau) = \widehat{\text{WADD}_{\bm{S}}}(\tau) \geq \widetilde{\text{WADD}_{\bm{S}}}(\tau). 
\end{flalign}
This holds for any trajectory ${\bm{S}}$. It then follows that
\begin{flalign}\label{eq:mixsig}
	\text{WADD}(\tau) &= \sup_{\nu\geq0}\sup_{\Omega_{\bm{S}}}\sup_{\bm{S}}\text{esssup}\mE^{{\bm{S}},\nu}_{\Omega_{\bm{S}}}[(\tau-\nu)^+|\mathbf X^n[1,\nu-1]]\nn\\&\geq \sup_{\nu\geq0}\sup_{\bm{S}} \text{esssup}\widetilde{\mE}^{\bm{S},\nu}[(\tau-\nu)^+|\widetilde{\mathbf X}^n[1,\nu-1]]\nn\\&=\widetilde{\text{WADD}}(\tau).
\end{flalign}
\textbf{Step 3.} The last step is to show that for any $\tau$ and any $\bm\beta$, $\widetilde{\text{WADD}}(\tau) \geq \widetilde{\text{WADD}}_{\bm{\beta}}(\tau)$.
Firstly, we will show that 
\begin{flalign}
	&\sup_{\bm{S}}\text{esssup} \widetilde{\mE}^{\bm{S}, \nu}\big[(\tau-\nu)^+|\widetilde{\mathbf X}^n[1,\nu-1]\big]\nn\\&\geq \text{esssup}\sup_{\bm{S}} \widetilde{\mE}^{\bm{S},\nu}\big[(\tau-\nu)^+|\widetilde{\mathbf X}^n[1,\nu-1]\big].
\end{flalign}
Let $c=\sup_{\bm{S}}\text{esssup} \widetilde{\mE}^{\bm{S},\nu}\big[(\tau-\nu)^+|\widetilde{\mathbf X}^n[1,\nu-1]\big]$. Denote by $\Lambda_{\bm{S}}$ the collection of all trajectory ${\bm{S}}$. For any ${\bm{S}}$, we have that
\begin{flalign}
	\widetilde{\mP}^\nu\Big(\widetilde{\mE}^{\bm{S},\nu}\big[(\tau-\nu)^+|\widetilde{\mathbf X}^n[1,\nu-1]\big]> c\Big) = 0.
\end{flalign}
Since $\Lambda_{\bm{S}}$ is countable, it then follows that 
\begin{flalign}
	&\widetilde{\mP}^\nu\Big(\sup\limits_{\bm{S}}\widetilde{\mE}^{\bm{S},\nu}\big[(\tau-\nu)^+|\widetilde{\mathbf X}^n[1,\nu-1]\big]> c\Big)\nn\\& \leq \widetilde{\mP}^\nu\Big(\cup_{{\bm{S}}\in\Lambda_{\bm{S}}}\big\{\widetilde{\mE}^{\bm{S},\nu}\big[(\tau-\nu)^+|\widetilde{\mathbf X}^n[1,\nu-1]\big]> c\big\}\Big) = 0.\nn
\end{flalign}
Therefore,
\begin{flalign}\label{eq:supp2}
	c&=\sup\limits_{{\bm{S}}} \text{esssup} \widetilde{\mE}^{\bm{S},\nu}\big[(\tau-\nu)^+|\widetilde{\mathbf X}^n[1,\nu-1]\big] \nn\\&\geq \text{esssup}\sup\limits_{{\bm{S}}} \widetilde{\mE}^{\bm{S},\nu}\big[(\tau-\nu)^+|\widetilde{\mathbf X}^n[1,\nu-1]\big].
\end{flalign}

For any $T\geq\nu+1$, we have that
\begin{flalign}
	&\sup_{{\bm{S}}}\sum_{t=\nu+1}^T (t-\nu)\widetilde{\mP}^{S[\nu],...,S[T]}(\tau=t|\widetilde{\mathbf{X}}^n[1,\nu-1])\nn\\&\geq \sum_{t=\nu+1}^T (t-\nu)\sum_{\{S[\nu],...,S[T]\}\in \Lambda_{\bm{S}}^{\bigotimes (T-\nu+1)}} \beta_{S[\nu]}\times\cdots\times\beta_{S[T]}\widetilde{\mP}^{S[\nu],...,S[T]}(\tau=t|\widetilde{\mathbf{X}}^n[1,\nu-1])\nn\\&=\sum_{t=\nu+1}^T (t-\nu)\widetilde{\mP}^{\bm\beta,\nu}(\tau=t|\widetilde{\mathbf{X}}^n[1,\nu-1]),
\end{flalign}
where $\widetilde{\mP}^{S[\nu],...,S[T]}$ denotes the probability measure when the trajectory is ${\bm{S}}$, the observations from time $\nu$ to time $T$ are generated according to $\widetilde{{\mP}}^{S[\nu]},...,\widetilde{\mP}^{S[T]}$.
As $T\rightarrow\infty$, we have
\begin{flalign}\label{eq:ave2}
	&\sup_{\bm{S}}\widetilde{\mE}^{\bm{S},\nu}[(\tau-\nu)^+|\widetilde{\mathbf X}^n[1,\nu-1]]\nn\\&\geq
	\widetilde{\mE}^{\bm{\beta},\nu}[(\tau-\nu)^+|\widetilde{\mathbf X}^n[1,\nu-1]].
\end{flalign}

From \eqref{eq:supp2} and \eqref{eq:ave2}, we have that
\begin{flalign}\label{eq:second2}
	\widetilde{\text{WADD}}(\tau) &= \sup\limits_{\bm{S}}\text{esssup} \widetilde{\mE}^{\bm{S},\nu}[(\tau-\nu)^+|\widetilde{\mathbf X}^n[1,\nu-1]]\nn\\&\geq \text{esssup}\widetilde{\mE}^{\bm{\beta},\nu}[(\tau-\nu)^+|\widetilde{\mathbf X}^n[1,\nu-1]]\nn\\&=
	\widetilde{\text{WADD}}_{\bm\beta}(\tau).
\end{flalign}
Combining \eqref{eq:mixsig} and \eqref{eq:second2}, we have that for any $\tau$ and $\bm\beta$ 
\begin{flalign}\label{eq:lowerbound}
	\text{WADD}(\tau) \geq \widetilde{\text{WADD}}(\tau)\geq \widetilde{\text{WADD}}_{\bm\beta}(\tau).
\end{flalign}

For any $T\geq 1$, we have that 
\begin{flalign}
	&\inf_{\substack{\{\sigma_1^0,...,\sigma_T^0\}\\\in\msn^{\bigotimes T}}}\sum_{t = 1}^T t\mP^\infty_{\sigma_1^0,...,\sigma_T^0}(\tau = t) \nn\\& \leq \sum_{t=1}^T t\frac{1}{\mid \msn \mid^{T}}\sum\limits_{\substack{\{\sigma_1^0,...,\sigma_T^0\}\\\in\msn^{\bigotimes T}}}\mP^\infty_{\sigma_1^0,...,\sigma_T^0}(\tau=t) \nn\\&=\sum_{t=1}^T t\widetilde{\mP}^\infty(\tau=t).
\end{flalign}
As $T\rightarrow \infty$, we have that $\text{ARL}(\tau) \leq \widetilde{\text{ARL}}(\tau)$.

Therefore, for any stopping time $\tau$ satisfying $\text{ARL}(\tau) \geq \gamma$, it will also satisfy $\widetilde{\text{ARL}}(\tau)\geq \gamma$. We then have that $\mathcal{C}_\gamma \subseteq \widetilde{\mathcal{C}}_\gamma$. 

Since \eqref{eq:lowerbound} holds for any $\bm{\beta}$, it holds for $\bm{\beta^*}$. Problem \eqref{eq:fixb} is a classical QCD problem. From the asymptotic lower bound analysis in \cite{lai1998information}, we have that for large $\gamma$,
\begin{flalign}\label{eq:dyupp}
	\inf_{\tau\in\mathcal{C}_\gamma}\text{WADD}(\tau)\geq\inf_{\tau\in\widetilde{\mathcal{C}}_\gamma} \widetilde{\text{WADD}}_{\bm\beta^*}(\tau)\sim \frac{\log\gamma}{I_{\bm{\beta^*}}}(1+o(1)).
\end{flalign}

\section{Proof of Lemma \ref{lemma:1}}\label{sec:lemma}
The minimization of $I_{\bm\beta}$ is to solve the following problem:
\begin{flalign}
	\inf_{\bm\beta} \quad I_{\bm\beta}\quad &\nn\\
	\text{s.t.} \quad-\beta_k &\leq 0, \ \text{for} \ k\in[1,K]\\
	\sum_{k=1}^K\beta_k -1 &= 0.\nn
\end{flalign}
This is a convex optimization problem with linear constraints. Define the Lagrange function $L(\bm{\beta}, \eta, \bm{\mu})$:
\begin{flalign}
	L(\bm{\beta}, \eta, \bm{\mu}) = I_{\bm\beta} + \eta \Big(\sum_{k=1}^K \beta_k-1\Big) -\sum_{k=1}^{K}\mu_k\beta_k.
\end{flalign}
The minimizer $\bm{\beta^*}$ satisfies the Karush–Kuhn–Tucker(KKT) conditions: $\mu_k,\ \beta_k^* \geq 0,\ \mu_k \beta_k^* = 0,\ \sum_{k=1}^K \beta_k^* -1 = 0$ and 
\begin{flalign}
	\frac{\partial L}{\partial \beta_k}|_{\bm{\beta^*}} = \widetilde {\mE}^k\Big[\log\frac{\widetilde{\mP}^{\bm{\beta^*}}(X^n)}{\widetilde{\mP}_0(X^n)}\Big] + 1 +\eta - \mu_k = 0,
\end{flalign}
where $\widetilde {\mE}^k$ denotes the expectation under the distribution $\widetilde \mP^k$.

When $\beta_k^* > 0$, we have $\mu_k = 0$. Therefore, for any $k, k' \in \mathcal{K}$ with $\beta_k^* , \beta_{k'}^*> 0$, we have 
\begin{flalign}
	\widetilde {\mE}^k\Big[\log\frac{\widetilde{\mP}^{\bm{\beta^*}}(X^n)}{\widetilde{\mP}_0(X^n)}\Big] = \widetilde {\mE}^{k'}\Big[\log\frac{\widetilde{\mP}^{\bm{\beta^*}}(X^n)}{\widetilde{\mP}_0(X^n)}\Big] = -(1+\eta).\nn
\end{flalign}
The set $\mathcal{K}$ can be divided into two disjoint parts $\mathcal{K}_1$ and $\mathcal{K}_2$. All $k$ in $\mathcal{K}_1$ satisfy $\beta_k^*>0$ while all $k$ in $\mathcal{K}_2$ have $\beta_k^*=0$. 
We have that
\begin{flalign}\label{eq:interior}
	I_{\bm{\beta^*}} &= \sum_{k=1}^K \beta_k^* \widetilde {\mE}^k\Big[\log\frac{\widetilde{\mP}^{\bm{\beta^*}}(X^n)}{\widetilde{\mP}_0(X^n)}\Big]\nn\\&= \sum_{k\in K_1}\beta_k^* \widetilde {\mE}^k\Big[\log\frac{\widetilde{\mP}^{\bm{\beta^*}}(X^n)}{\widetilde{\mP}_0(X^n)}\Big]\nn\\& = \widetilde {\mE}^k\Big[\log\frac{\widetilde{\mP}^{\bm{\beta^*}}(X^n)}{\widetilde{\mP}_0(X^n)}\Big], \ \text{for all} \ k \in \mathcal{K}_1.
\end{flalign}
For all $k \in \mathcal{K}_2$, $\beta_k^*=0$. By the KKT conditions, we have that $\mu_k \geq 0$.
Therefore, for any $k\in \mathcal{K}_2$,
$
	\widetilde {\mE}^k\Big[\log\frac{\widetilde{\mP}^{\bm{\beta^*}}(X^n)}{\widetilde{\mP}_0(X^n)}\Big] + 1 +\eta = \mu_k \geq 0.
$
We then have that for any $k \in \mathcal{K}_2$, 
\begin{flalign}\label{eq:boundary}
	\widetilde {\mE}^k\Big[\log\frac{\widetilde{\mP}^{\bm{\beta^*}}(X^n)}{\widetilde{\mP}_0(X^n)}\Big] \geq I_{\bm{\beta^*}}.
\end{flalign}

\section{Proof of Theorem \ref{theorem:4}}\label{sec:upperwadds}
Due to the fact that the test statistic $\max_{1\leq k\leq t+1}\sum_{i=k}^t \ell_{\bm\beta^*}(X^n_i)$ has initial value 0 and remains non-negative, the delay is largest when the change happens at $\nu = 0$. Therefore, for any ${\bm{S}}$, we have that
\begin{flalign}
	\text{WADD}_{\bm{S}}(T_{\bm{\beta^*}}) &= \sup_{\nu\geq 0}\sup_{\Omega_{\bm{S}}}\text{esssup}\mE_{\Omega_{\bm{S}}}^{{\bm{S}},\nu}\big[(T_{\bm{\beta^*}}-\nu)^+\mid \mathbf X^n[1,\nu-1]\big]\nn\\&=\sup_{\Omega_{\bm{S}}}\mE^{{\bm{S}},0}_{\Omega_{\bm{S}}}[T_{\bm{\beta^*}}].
\end{flalign}
For any $T\geq\nu+1$, we have that
\begin{flalign}
	&\sup_{\substack{\{\sigma_1^{S[1]},\cdots,\sigma_T^{S[T]}\}\\\in\mathcal S_{n,S[1]} \times,\cdots,\times \mathcal S_{n,S[T]}}}\sum_{t=1}^T t\mP_{\sigma_1^{S[1]},\cdots,\sigma_T^{S[T]}}^{{\bm{S}},0}(T_{\bm{\beta^*}}=t)\nn\\&
	= \sum_{t=1}^T t\frac{1}{\mid \mathcal S_{n,S[1]} \mid\times\cdots \times \mid\mathcal S_{n,S[T]}\mid}\sum\limits_{\substack{\{\sigma_1^{S[1]},\cdots,\sigma_T^{S[T]}\}\\\in\mathcal S_{n,S[1]} \times,\cdots,\times \mathcal S_{n,S[T]}}}\mP_{\sigma_1^{S[1]},\cdots,\sigma_T^{S[T]}}^{{\bm{S}},0}(T_{\bm{\beta^*}}=t)\nn\\&
	=\sum_{t=1}^T t\widetilde{\mP}^{\bm{S},0}(T_{\bm{\beta^*}}=t).
\end{flalign}
As $T\rightarrow\infty$, we have that
\begin{flalign}\label{eq:equal}
	\sup_{\Omega_{\bm{S}}}\mE^{{\bm{S}},0}_{\Omega_{\bm{S}}}[T_{\bm{\beta^*}}] = \widetilde{\mE}^{\bm{S},0}[T_{\bm{\beta^*}}]= \widetilde{\text{WADD}_{\bm{S}}}(T_{\bm{\beta^*}}). 
\end{flalign}
For any ${\bm{S}}$, we have $\text{WADD}_{\bm{S}}(T_{\bm{\beta^*}}) = \widetilde{\text{WADD}_{\bm{S}}}(T_{\bm{\beta^*}})$. Therefore, $\text{WADD}(T_{\bm{\beta^*}}) = \widetilde{\text{WADD}}(T_{\bm{\beta^*}})$ by taking sup over ${\bm{S}}$ on both sides.
It then follows that
\begin{flalign}
	&\text{WADD}(T_{\bm{\beta^*}})=\widetilde{\text{WADD}}(T_{\bm{\beta^*}}) = \sup_{\bm{S}}\widetilde{\mE}^{\bm{S},0}[T_{\bm{\beta^*}}].
\end{flalign}
Let $0<\epsilon<I_{\bm{\beta^*}}$ and $n_b = \frac{b}{I_{\bm{\beta^*}}-\epsilon}$. For any trajectory ${\bm{S}}$, from the sum-integral inequality, we have that
\begin{flalign}\label{eq:part1}
	\widetilde{\mE}^{\bm{S},0}\Big[\frac{T_{\bm{\beta^*}}}{n_b}\Big] &= \int_0^\infty \widetilde{\mP}^{\bm{S},0}\Big(\frac{T_{\bm{\beta^*}}}{n_b}>x\Big)\mathrm{d}x \nn\\&\leq \sum^\infty_{t=1}\widetilde{\mP}^{\bm{S},0}(T_{\bm{\beta^*}}>tn_b) + 1.
\end{flalign}
For any ${\bm{S}}$, we have that
\begin{flalign}
	&\widetilde{\mP}^{\bm{S},0}(T_{\bm{\beta^*}}>tn_b)\nn\\&= \widetilde{\mP}^{\bm{S},0}\bigg(\max_{1\leq k\leq tn_b}\max_{1\leq i\leq k}\sum_{j=i}^{k}\ell_{\bm{\beta^*}}(X^n_j)<b\bigg)\nn\\&\leq \widetilde{\mP}^{\bm{S},0}\bigg(\max_{1\leq i\leq mn_b}\sum_{j=i}^{mn_b}\ell_{\bm{\beta^*}}(X^n_j)<b, \forall m \in [t]\bigg)\nn\\&\leq \widetilde{\mP}^{\bm{S},0}\bigg(\sum_{j=(m-1)n_b+1}^{mn_b}\ell_{\bm{\beta^*}}(X^n_j)<b, \forall m \in [t]\bigg)\nn\\& = \widetilde{\mP}^{\bm{S},0}\Bigg(\frac{\sum\limits_{j=(m-1)n_b+1}^{mn_b}\ell_{\bm{\beta^*}}(X^n_j)}{n_b}<I_{\bm{\beta^*}}-\epsilon, \forall m \in [t]\Bigg)\nn\\&= \prod_{m=1}^t \widetilde{\mP}^{\bm{S},0}\Bigg(\frac{\sum\limits_{j=(m-1)n_b+1}^{mn_b}\ell_{\bm{\beta^*}}(X^n_j)}{n_b}<I_{\bm{\beta^*}}-\epsilon\Bigg).
\end{flalign}
It then follows that 
\begin{flalign}
	&\sup_{\bm{S}}\sum_{t=1}^{\infty}\widetilde{\mP}^{\bm{S},0}(T_{\bm{\beta^*}}>tn_b) \nn\\&\leq \sup_{\bm{S}}\sum_{t=1}^\infty \prod_{m=1}^t\widetilde{\mP}^{\bm{S},0}\bigg(\frac{\sum\limits_{j=(m-1)n_b+1}^{mn_b}\ell_{\bm{\beta^*}}(X^n_j)}{n_b}<I_{\bm{\beta^*}}-\epsilon\bigg).\nn
\end{flalign}
Then we will bound $\widetilde{\mP}^{\bm{S},0}\bigg(\frac{\sum\limits_{j=(m-1)n_b+1}^{mn_b}\ell_{\bm{\beta^*}}(X^n_j)}{n_b}<I_{\bm{\beta^*}}-\epsilon\bigg)$.

Let $I_{{\bm{S}}_m} = \widetilde{\mE}^{\bm{S},0}\bigg[\frac{\sum\limits_{j=(m-1)n_b+1}^{mn_b}\ell_{\bm\beta^*}(X^n_j)}{n_b}\bigg]$. From \eqref{eq:interior} and \eqref{eq:boundary}, we have that
\begin{flalign}
	I_{{\bm{S}}_m} &= \widetilde{\mE}^{\bm{S},0}\Bigg[\frac{\sum\limits_{j=(m-1)n_b+1}^{mn_b}\ell_{\bm\beta^*}(X^n_j)}{n_b}\Bigg] \nn\\&= \sum\limits_{j=(m-1)n_b+1}^{mn_b}\widetilde{\mE}^{S[j]}\bigg[\frac{\ell_{\bm\beta^*}(X^n_j)}{n_b}\bigg] \nn\\&= \frac{1}{n_b}\sum\limits_{j=(m-1)n_b+1}^{mn_b}\widetilde{\mE}^{S[j]}\big[\ell_{\bm\beta^*}(X^n_j)\big]\geq I_{\bm{\beta^*}}.
\end{flalign}
It then follows that for any ${\bm{S}}$ and $m$
\begin{flalign}
	&\widetilde{\mP}^{\bm{S},0}\Bigg(\frac{\sum\limits_{j=(m-1)n_b+1}^{mn_b}\ell_{\bm{\beta^*}}(X^n_j)}{n_b}<I_{\bm{\beta^*}}-\epsilon\Bigg)\nn\\&\leq \widetilde{\mP}^{\bm{S},0}\Bigg(\frac{\sum\limits_{j=(m-1)n_b+1}^{mn_b}\ell_{\bm{\beta^*}}(X^n_j)}{n_b}<I_{{\bm{S}}_m}-\epsilon\Bigg)\nn\\&\leq \widetilde{\mP}^{\bm{S},0}\Bigg(\Bigg|\frac{\sum\limits_{j=(m-1)n_b+1}^{mn_b}\ell_{\bm{\beta^*}}(X^n_j)}{n_b}-I_{{\bm{S}}_m}\Bigg|>\epsilon\Bigg).
\end{flalign}
Assume that $\max_{k\in[1,K]}\widetilde{\mE}^k\bigg[\ell_{\bm{\beta^*}}(X^n)^2\bigg]<\infty.$
Let $\sigma^2=\max_{k\in [1,K]}\text{Var}_{\widetilde{\mP}^k}(\ell_{\bm{\beta^*}}(X^n))$ where $\text{Var}_{\widetilde{\mP}^k}$ denotes the variance under the distribution $\widetilde{\mP}^k$.
By Chebychev's inequality, 
\begin{flalign}\label{eq:part2}
	&\widetilde{\mP}^{\bm{S},0}\Bigg(\Bigg|\frac{\sum\limits_{j=(m-1)n_b+1}^{mn_b}\ell_{\bm{\beta^*}}(X^n_j)}{n_b}-I_{{\bm{S}}_m}\Bigg|>\epsilon\Bigg) \nn\\&\leq \text{Var}_{\widetilde{\mP}^{\bm{S}}}\Bigg(\frac{\sum\limits_{j=(m-1)n_b+1}^{mn_b}\ell_{\bm{\beta^*}}(X^n_j)}{n_b}\Bigg)\frac{1}{\epsilon^2}\nn\\& = \frac{1}{\epsilon^2n_b^2}\sum_{j=(m-1)n_b+1}^{mn_b}\text{Var}_{\widetilde{\mP}^{S[j]}}(\ell_{\bm{\beta^*}}(X^n_j))\nn\\&\leq \frac{\sum\limits_{j=(m-1)n_b+1}^{mn_b}\sigma^2}{n_b^2\epsilon^2} = \frac{\sigma^2}{n_b\epsilon^2}.
\end{flalign}
Let $\delta = \frac{\sigma^2}{n_b\epsilon^2}$. From \eqref{eq:part1} and \eqref{eq:part2}, we have that
\begin{flalign}
	\sup_{\bm{S}}\widetilde{\mE}^{\bm{S},0}\Big[\frac{T_{\bm{\beta^*}}}{n_b}\Big]&\leq 1 + \sup_{\bm{S}}\sum^\infty_{t=1}\widetilde{\mP}^{\bm{S},0}(T_{\bm{\beta^*}}>tn_b)\nn\\&\leq 1+ \sum_{t=1}^{\infty}(\frac{\sigma^2}{n_b\epsilon^2})^t\nn\\& = 1 + \sum_{t=1}^{\infty}\delta^t = \frac{1}{1-\delta}.
\end{flalign}
Therefore, we have 
\begin{flalign}\label{eq:upper}
	\sup_{\bm{S}}\widetilde{\mE}^{\bm{S},0}\Big[T_{\bm{\beta^*}}\Big]\leq \frac{b}{(I_{\bm{\beta^*}}-\epsilon)(1-\delta)}.
\end{flalign}
\eqref{eq:upper} holds for all $\epsilon$. It then follows that as $b \rightarrow \infty$, 
\begin{flalign}\label{eq:dylo}
	\text{WADD}(T_{\bm{\beta^*}}) = \sup_{\bm{S}}\widetilde{\mE}^{\bm{S},0}\Big[T_{\bm{\beta^*}}\Big] \leq \frac{b}{I_{\bm{\beta^*}}}(1+o(1)).
\end{flalign}
For the ARL lower bound, for any $T\geq 1$, we have that 
\begin{flalign}
	&\inf_{\substack{\{\sigma_1^0,...,\sigma_T^0\}\\\in\msn^{\bigotimes T}}}\sum_{t = 1}^T t\mP^\infty_{\sigma_1^0,...,\sigma_T^0}(T_{\bm\beta^*} = t) \nn\\& = \sum_{t=1}^T t\frac{1}{\mid \msn \mid^{T}}\sum\limits_{\substack{\{\sigma_1^0,...,\sigma_T^0\}\\\in\msn^{\bigotimes T}}}\mP^\infty_{\sigma_1^0,...,\sigma_T^0}(T_{\bm\beta^*}=t) \nn\\&=\sum_{t=1}^T t\widetilde{\mP}^\infty(T_{\bm\beta^*}=t).
\end{flalign}
As $T \rightarrow \infty$, we have that 
$
	\text{WARL}(T_{\bm\beta^*}) = \widetilde{\text{ARL}}(T_{\bm\beta^*}).
$
$T_{\bm\beta^*}$ is the CuSum algorithm for a simple QCD problem with pre-change distribution $\widetilde{\mP}_0$ and post-change distribution $\widetilde{\mP}^{\bm{\beta^*}}$. From the optimal property of CuSum algorithm in \cite{lorden1971procedures} and \cite{moustakides1986optimal}, we have that when $b = \log\gamma$,
$
	\text{WARL}(T_{\bm\beta^*}) = \widetilde{\text{ARL}}(T_{\bm\beta^*}) \geq \gamma.
$
\normalem
\bibliographystyle{ieeetr}
\bibliography{QCD}

\begin{thebibliography}{10}

\bibitem{sun2021dynamic}
Z.~Sun and S.~Zou, ``Quickest dynamic anomaly detection in anonymous
  heterogeneous sensor networks,'' in {\em Proc. IEEE Int. Symp. Inf. Theory
  (ISIT)}, pp.~106--111, 2021.

\bibitem{humphreys2008assessing}
T.~E. Humphreys, B.~M. Ledvina, M.~L. Psiaki, B.~W. O'Hanlon, P.~M. Kintner,
  {\em et~al.}, ``Assessing the spoofing threat: Development of a portable gps
  civilian spoofer,'' in {\em Proceedings of the 21st International Technical
  Meeting of the Satellite Division of The Institute of Navigation (ION GNSS
  2008)}, pp.~2314--2325, 2008.

\bibitem{keller2009identity}
L.~Keller, M.~J. Siavoshani, C.~Fragouli, K.~Argyraki, and S.~Diggavi,
  ``Identity aware sensor networks,'' in {\em Proc. Int. Conf. Commun.,
  Computing Control Appl.}, pp.~2177--2185, IEEE, 2009.

\bibitem{anonymous}
W.~N. Chen and I.~H. Wang, ``Anonymous heterogeneous distributed detection:
  Optimal decision rules, error exponents, and the price of anonymity,'' {\em
  IEEE Trans. Inform. Theory}, vol.~65, no.~11, pp.~7390--7406, 2019.

\bibitem{li2022bandwidth}
W.~Li and Y.~Huang, ``Bandwidth-constrained distributed quickest change
  detection in heterogeneous sensor networks: Anonymous vs non-anonymous
  settings,'' {\em arXiv preprint arXiv: 2202.02697}, 2022.

\bibitem{stefano2019unlabeled}
S.~Marano and P.~K. Willett, ``Algorithms and fundamental limits for unlabeled
  detection using types,'' {\em IEEE Trans. Signal Proc.}, vol.~67, no.~8,
  pp.~2022--2035, 2019.

\bibitem{stefano2020bits}
S.~Marano and P.~Willett, ``Making decisions by unlabeled bits,'' {\em IEEE
  Trans. Signal Proc.}, vol.~68, pp.~2935--2947, 2020.

\bibitem{unnikrishnan2018unlabeled}
J.~Unnikrishnan, S.~Haghighatshoar, and M.~Vetterli, ``Unlabeled sensing with
  random linear measurements,'' {\em IEEE Trans. Inform. Theory}, vol.~64,
  no.~5, pp.~3237--3253, 2018.

\bibitem{haghighatshoar2017signal}
S.~Haghighatshoar and G.~Caire, ``Signal recovery from unlabeled samples,''
  {\em IEEE Trans. Signal Proc.}, vol.~66, no.~5, pp.~1242--1257, 2017.

\bibitem{abid2017linear}
A.~Abid, A.~Poon, and J.~Zou, ``Linear regression with shuffled labels,'' {\em
  arXiv preprint arXiv:1705.01342}, 2017.

\bibitem{emiya2014compressed}
V.~Emiya, A.~Bonnefoy, L.~Daudet, and R.~Gribonval, ``Compressed sensing with
  unknown sensor permutation,'' in {\em Proc. IEEE Int. Conf. Acoust. Speech
  Signal Process. (ICASSP)}, pp.~1040--1044, IEEE, 2014.

\bibitem{liu2018signal}
Z.~Liu and J.~Zhu, ``Signal detection from unlabeled ordered samples,'' {\em
  IEEE Commun. Lett.}, vol.~22, no.~12, pp.~2431--2434, 2018.

\bibitem{pananjady2017linear}
A.~Pananjady, M.~J. Wainwright, and T.~A. Courtade, ``Linear regression with
  shuffled data: Statistical and computational limits of permutation
  recovery,'' {\em IEEE Trans. Inform. Theory}, vol.~64, no.~5, pp.~3286--3300,
  2017.

\bibitem{elhami2017unlabeled}
G.~Elhami, A.~Scholefield, B.~B. Haro, and M.~Vetterli, ``Unlabeled sensing:
  Reconstruction algorithm and theoretical guarantees,'' in {\em Proc. IEEE
  Int. Conf. Acoust. Speech Signal Process. (ICASSP)}, pp.~4566--4570, Ieee,
  2017.

\bibitem{lu2008theory}
Y.~M. Lu and M.~N. Do, ``A theory for sampling signals from a union of
  subspaces,'' {\em IEEE Trans. Signal Proc.}, vol.~56, no.~6, pp.~2334--2345,
  2008.

\bibitem{wang2018signal}
G.~Wang, J.~Zhu, R.~S. Blum, P.~Willett, S.~Marano, V.~Matta, and P.~Braca,
  ``Signal amplitude estimation and detection from unlabeled binary quantized
  samples,'' {\em IEEE Trans. Signal Proc.}, vol.~66, no.~16, pp.~4291--4303,
  2018.

\bibitem{sun2020tspanonymous}
Z.~Sun, S.~Zou, R.~Zhang, and Q.~Li, ``Quickest change detection in anonymous
  heterogeneous sensor networks,'' {\em IEEE Trans. Signal Proc.}, vol.~70,
  pp.~1041--1055, 2022.

\bibitem{rovatsos2021quickest}
G.~Rovatsos, G.~V. Moustakides, and V.~V. Veeravalli, ``Quickest detection of
  moving anomalies in sensor networks,'' {\em IEEE J. Sel. Areas Inf. Theory},
  vol.~2, no.~2, pp.~762--773, 2021.

\bibitem{tartakovsky2004change}
A.~G. Tartakovsky and V.~V. Veeravalli, ``Change-point detection in
  multichannel and distributed systems,'' {\em Applied Sequential
  Methodologies: Real-World Examples with Data Analysis}, vol.~173,
  pp.~339--370, 2004.

\bibitem{tartakovsky2006novel}
A.~G. Tartakovsky, B.~L. Rozovskii, R.~B. Blazek, and H.~Kim, ``A novel
  approach to detection of intrusions in computer networks via adaptive
  sequential and batch-sequential change-point detection methods,'' {\em IEEE
  Trans. Signal Proc.}, vol.~54, no.~9, pp.~3372--3382, 2006.

\bibitem{mei2010efficient}
Y.~Mei, ``Efficient scalable schemes for monitoring a large number of data
  streams,'' {\em Biometrika}, vol.~97, no.~2, pp.~419--433, 2010.

\bibitem{xie2013sequential}
Y.~Xie and D.~Siegmund, ``Sequential multi-sensor change-point detection,''
  {\em Ann. Statist.}, pp.~670--692, 2013.

\bibitem{fellouris2016second}
G.~Fellouris and G.~Sokolov, ``Second-order asymptotic optimality in
  multisensor sequential change detection,'' {\em IEEE Trans. Inform. Theory},
  vol.~62, no.~6, pp.~3662--3675, 2016.

\bibitem{raghavan2010quickest}
V.~Raghavan and V.~V. Veeravalli, ``Quickest change detection of a {M}arkov
  process across a sensor array,'' {\em IEEE Trans. Inform. Theory}, vol.~56,
  no.~4, pp.~1961--1981, 2010.

\bibitem{hadjiliadis2009one}
O.~Hadjiliadis, H.~Zhang, and H.~V. Poor, ``One shot schemes for decentralized
  quickest change detection,'' {\em IEEE Trans. Inform. Theory}, vol.~55,
  no.~7, pp.~3346--3359, 2009.

\bibitem{ludkovski2012bayesian}
M.~Ludkovski, ``Bayesian quickest detection in sensor arrays,'' {\em Seq.
  Anal.}, vol.~31, no.~4, pp.~481--504, 2012.

\bibitem{zou2020dynamic}
S.~{Zou}, V.~V. {Veeravalli}, J.~{Li}, and D.~{Towsley}, ``Quickest detection
  of dynamic events in networks,'' {\em IEEE Trans. Inform. Theory}, vol.~66,
  no.~4, pp.~2280--2295, 2020.

\bibitem{veeravalli2001decentralized}
V.~V. Veeravalli, ``Decentralized quickest change detection,'' {\em IEEE Trans.
  Inform. Theory}, vol.~47, no.~4, pp.~1657--1665, 2001.

\bibitem{tartakovsky2008distributed}
A.~G. Tartakovsky and V.~V. Veeravalli, ``Asymptotically optimal quickest
  change detection in distributed sensor systems,'' {\em Seq. Anal.}, vol.~27,
  no.~4, pp.~441--475, 2008.

\bibitem{zou2019distributed}
S.~Zou, V.~V. Veeravalli, J.~Li, D.~Towsley, and A.~Swami, ``Distributed
  quickest detection of significant events in networks,'' in {\em Proc. IEEE
  Int. Conf. Acoust. Speech Signal Process. (ICASSP)}, pp.~8454--8458, 2019.

\bibitem{xie2021sequential}
L.~Xie, S.~Zou, Y.~Xie, and V.~V. Veeravalli, ``Sequential (quickest) change
  detection: Classical results and new directions,'' {\em IEEE J. Sel. Areas
  Inf. Theory}, vol.~2, no.~2, pp.~494--514, 2021.

\bibitem{rovatsos2019dynamic}
G.~Rovatsos, G.~Moustakides, and V.~Veeravalli, ``Quickest detection of a
  dynamic anomaly in a sensor network,'' in {\em Proc. Asilomar Conf. Signals,
  Systems and Computers}, pp.~98--102, 2019.

\bibitem{georgios2020movinganomaly}
G.~{ Rovatsos}, S.~{Zou}, and V.~V. {Veeravalli}, ``Sequential algorithms for
  moving anomaly detection in networks,'' {\em Seq. Anal.}, vol.~39, no.~1,
  pp.~6--31, 2020.

\bibitem{zou2018quickest}
S.~Zou, G.~Fellouris, and V.~V. Veeravalli, ``Quickest change detection under
  transient dynamics: Theory and asymptotic analysis,'' {\em IEEE Trans.
  Inform. Theory}, vol.~65, no.~3, pp.~1397--1412, 2018.

\bibitem{zhang2019quickest}
R.~Zhang, R.~Yao, Y.~Xie, and F.~Qiu, ``Quickest detection of cascading
  failure,'' {\em arXiv preprint arXiv:1911.05610}, 2019.

\bibitem{siegmund1995using}
D.~Siegmund and E.~S. Venkatraman, ``Using the generalized likelihood ratio
  statistic for sequential detection of a change-point,'' {\em Ann. Statist.},
  pp.~255--271, 1995.

\bibitem{lai1998information}
T.~L. Lai, ``Information bounds and quick detection of parameter changes in
  stochastic systems,'' {\em IEEE Trans. Inform. Theory}, vol.~44, no.~7,
  pp.~2917--2929, 1998.

\bibitem{banerjee2015composite}
T.~{Banerjee} and V.~V. {Veeravalli}, ``Data-efficient minimax quickest change
  detection with composite post-change distribution,'' {\em IEEE Trans. Inform.
  Theory}, vol.~61, no.~9, pp.~5172--5184, 2015.

\bibitem{lorden1971procedures}
G.~Lorden, ``Procedures for reacting to a change in distribution,'' {\em Ann.
  Math. Statist.}, vol.~42, no.~6, pp.~1897--1908, 1971.

\bibitem{williams_1991}
D.~Williams, {\em Probability with Martingales}.
\newblock Cambridge University Press, 1991.

\bibitem{moustakides1986optimal}
G.~V. Moustakides, ``Optimal stopping times for detecting changes in
  distributions,'' {\em Ann. Statist.}, vol.~14, no.~4, pp.~1379--1387, 1986.

\end{thebibliography}
\end{document}